\documentclass[10pt,letterpaper]{article}
\usepackage{cite}
\usepackage[utf8]{inputenc}
\usepackage[hidelinks]{hyperref}
\usepackage{amssymb}
\usepackage{amsthm}
\usepackage{amsmath}
\usepackage{enumerate}
\usepackage{dsfont}
\usepackage{booktabs}
\usepackage{thm-restate}
\usepackage{cleveref}
\usepackage{graphicx}
\usepackage{color}
\usepackage{tikz}
\usepackage[disable]{todonotes}
\usepackage{comment}
\usepackage{placeins}
\usepackage[]{algorithm2e}
\usepackage{xspace}
\usepackage{subcaption}
\usepackage[appendix=inline]{apxproof}
\usepackage{authblk}
\usepackage[hidelinks]{hyperref}
\newcommand*{\email}[1]{\href{mailto:#1}{\nolinkurl{#1}} } 

\newenvironment{claimproof}[1]{\par\noindent\underline{Proof:}\space#1}{\hfill $\blacksquare$}

\newif\iflong\longfalse

\newcommand{\Oh}{\mathcal{O}}

\newcommand{\W}[1]{\ensuremath{\mathrm{W}[#1]}\xspace}

\newcommand\NP{\ensuremath{\mathrm{NP}}\xspace}

\newcommand\XP{\ensuremath{\mathrm{XP}}\xspace}

\newcommand\coNP{\ensuremath{\mathrm{coNP}}\xspace}

\newcommand\FPT{\ensuremath{\mathrm{FPT}}\xspace}

\newcommand{\modid}[2]{\ensuremath{{#1}\smash{\left[{#2}\right]^{\circ}}}}

\newcommand{\PCA}{\textsc{Periodic Character Alignment}\xspace}
\newcommand{\ST}{\textsc{EPG Short Traversal}\xspace}
\newcommand{\SG}{\textsc{EPG Subgraph}\xspace}
\newcommand{\SGfree}{\textsc{EPG Subgraph-Free}\xspace}
\newcommand{\Minor}{\textsc{EPG Minor}\xspace}
\newcommand{\Minorfree}{\textsc{EPG Minor-Free}\xspace}
\newcommand{\MCC}{\textsc{Multicolored Clique}\xspace}
\newcommand{\MCPCA}{\textsc{Multicolored PCA}\xspace}

{\itshape}{\rmfamily}
\theoremstyle{plain}
\newtheorem{theorem}{Theorem}
\newtheorem{lemma}{Lemma}
\newtheorem{corollary}{Corollary}

\theoremstyle{remark}
\newtheorem{remark}{Remark}
\theoremstyle{definition}
\newtheorem{definition}{Definition}

\newcommand{\LCM}{\ensuremath{\operatorname{lcm}}\xspace}

\newcommand{\prob}[3]{\begin{quote}  \textsc{#1}\\  \textbf{Input:} #2\\  \textbf{Question:} #3\end{quote}}

\begin{document}

\title{Multi-Parameter Analysis of Finding Minors and Subgraphs in Edge Periodic Temporal Graphs}
\author[1]{Emmanuel Arrighi\footnote{Supported by Research Council of Norway (no.~274526) and IS-DAAD (no.~309319)}}
\author[2]{Niels Grüttemeier}
\author[2]{Nils Morawietz\footnote{Supported by Deutsche Forschungsgemeinschaft, project OPERAH, KO~3669/5-1.}}
\author[2]{Frank~Sommer\footnote{Supported by Deutsche Forschungsgemeinschaft, project EAGR, {KO~3669/{6-1}}.}}
\author[3]{Petra Wolf\footnote{Supported by Deutsche Forschungsgemeinschaft, project	FE 560/9-1 and DAAD PPP (no.~57525246).}}
\affil[1]{Universitetet i Bergen, Norway, \texttt{emmanuel.arrighi@uib.no}}
\affil[2]{Philipps-Universität Marburg, Germany \texttt{\{niegru, morawietz, fsommer\}@informatik.uni-marburg.de}}
\affil[3]{Universität Trier, Germany \texttt{wolfp@informatik.uni-trier.de}}

\date{}
\maketitle


\begin{abstract}  
  We study the computational complexity of determining structural properties of edge periodic temporal graphs (EPGs). 
  EPGs are time-varying graphs that compactly represent periodic behavior of components of a dynamic network, for example, train schedules on a rail network.
  In EPGs, for each edge~$e$ of the graph, a binary string $s_e$ determines in which time steps the edge is present, namely $e$ is present in time step $t$ if and only if $s_e$ contains a $1$ at position $t \mod |s_e|$. 
  Due to this periodicity, EPGs serve as very compact representations of complex periodic systems and can even be exponentially smaller than classic temporal graphs representing one period of the same system, as the latter contain the whole sequence of graphs explicitly.
  In this paper, we study the computational complexity of fundamental questions of the new concept of EPGs such as what is the shortest traversal time between two vertices; is there a time step in which the graph (1) is minor-free; (2) contains a minor; (3) is subgraph-free; (4) contains a subgraph; with respect to a given minor or subgraph.
  We give a detailed parameterized analysis for multiple combinations of parameters for the problems stated above including several parameterized algorithms.  
\end{abstract}

\section{Introduction}
In general, a \emph{time-varying graph} describes a graph that changes over time. For most applications, this change is limited to the availability or weight of edges, meaning that edges are only present at certain time steps or the time needed to cross an edge changes over time. They are of great interest in the area of \emph{dynamic networks}~\cite{DBLP:journals/paapp/CasteigtsFQS12,ganguly2009dynamics,holme2015modern,holme2012temporal} such as \emph{mobile ad hoc networks}~\cite{DBLP:journals/comsur/Zhang06} and \emph{vehicular networks}~\cite{DBLP:conf/edbt/DingYQ08,DBLP:journals/networks/Berman96} as in those networks, the topology naturally changes over time. There are plenty of representations for time-varying graphs in the literature which are not equivalent in general, see~\cite{DBLP:journals/paapp/CasteigtsFQS12,casteigts:hal-00865762,casteigts:hal-00865764} for some overview.
In general, a time-varying graph $\mathcal{G}$ consists of an underlying graph $G$ and functions describing how the availability or weights of edges change over time. Thereby, settings with \emph{discrete} and \emph{continuous} time steps are considered~\cite{DBLP:journals/paapp/CasteigtsFQS12,DBLP:journals/cacm/MichailS18,DBLP:journals/jcss/KempeKK02,DBLP:journals/algorithmica/MertziosMS19}. In this work, we only deal with the discrete time setting. Usually, in the field of time-varying graphs, for each time step $t$ of the \emph{lifetime} of the graph, the \emph{snapshot} graph~$\mathcal{G}(t)$, i.e., the graph present in time step $t$, is explicitly given in the input~\cite{bhadra2003complexity,DBLP:journals/tcs/MichailS16,wehmuth2015unifying}. This implies that the lifetime of the graph $\mathcal{G}$ is linear in the input size and further that the input is mostly dominated by the sequence of snapshot graphs $(\mathcal{G}(t))_t$ and not by the underlying graph $G$. We will call those time-varying graphs where the whole sequence of snapshot graphs is explicitly given \emph{temporal graphs}.

Knowing the whole sequence of snapshot graphs of the temporal graph requires a detailed knowledge of the usually complex system that is modelled by the graph. On the other hand, describing a system by its components is a natural concept in computer science~\cite{dRLP98,DBLP:conf/concur/JeckerM021,DBLP:journals/rsl/Pagin21} and requires only individual knowledge of the components.
In the context of time-varying graphs, this approach is realized by so called \emph{edge periodic (temporal) graphs}, EPGs for short, categorized as Class~8 in~\cite{DBLP:journals/paapp/CasteigtsFQS12} and considered for instance in~\cite{erlebach2020game,DBLP:conf/mfcs/MorawietzRW20,DBLP:conf/mfcs/Morawietz021}.
An edge periodic (temporal) graph $\mathcal{G} = (V, E, \tau)$ consists of an underlying graph $G = (V, E)$ and a function~$\tau$ that assigns each edge with a binary string, the \emph{edge label}, that indicates in which time step the edge is present. Thereby, the time step is considered modulo the length of the edge label. As the length of the edge labels can differ, the sequence of snapshot graphs only repeats after the least common multiple of the individual edge label lengths. Hence, an EPG can compactly represent an exponentially longer sequence of snapshot graphs without explicitly describing each snapshot graph individually. This implies that the lifetime can be exponentially in the input size. Fig.~\ref{fig:example} shows an example of an EPG together with some snapshot graphs.
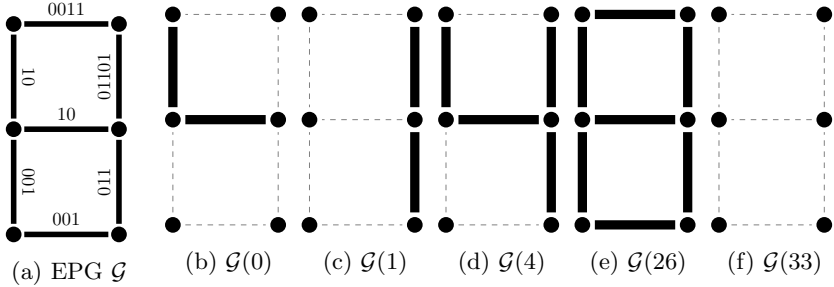
\begin{figure}[t]
    \centering
    \tikzset{%
    base/.pic = {%
        \begin{scope}[every node/.style={fill,circle,inner sep=3pt}]
            \node at (0,0) (-0) {};
            \node at (2,0) (-1) {};
            \node at (2,2) (-2) {};
            \node at (2,4) (-3) {};
            \node at (0,4) (-4) {};
            \node at (0,2) (-5) {};
        \end{scope}
    }
}
\begin{subfigure}{0.19\textwidth}
    \centering
    \scalebox{0.7}{%
    \begin{tikzpicture}[auto, line width=3pt, every node/.style={sloped}]
        \draw (0,0) pic[scale=1] (G) {base};
        \draw (G-0) to [edge label = 001] (G-1);
        \draw (G-1) to node () {011} (G-2);
        \draw (G-2) to node () {01101} (G-3);
        \draw (G-3) to node () {0011} (G-4);
        \draw (G-4) to node () {10} (G-5);
        \draw (G-5) to node () {001} (G-0);
        \draw (G-5) to node () {10} (G-2);
    \end{tikzpicture}
    }
    \caption{EPG $\mathcal{G}$}%
    \label{subfig:base}
\end{subfigure}
\begin{subfigure}{0.14\textwidth}
    \centering
    \scalebox{0.7}{%
    \begin{tikzpicture}[auto,line width=5pt]
        \draw (0,0) pic (G) {base};
        
        \draw (G-4) to (G-5);
        \draw[dashed,help lines] (G-3) to (G-4);
        \draw (G-2) to (G-5);
        \draw[dashed,help lines] (G-2) to (G-3);
        \draw[dashed,help lines] (G-1) to (G-2);
        \draw[dashed,help lines] (G-0) to (G-5);
        \draw[dashed,help lines] (G-0) to (G-1);
        
    \end{tikzpicture}
    }
    \caption{$\mathcal{G}(0)$}%
    \label{subfig:t0}
\end{subfigure}
\begin{subfigure}{0.14\textwidth}
    \centering
    \scalebox{0.7}{%
    \begin{tikzpicture}[auto,line width=5pt]
        \draw (0,0) pic (G) {base};
        
        \draw[dashed,help lines] (G-4) to (G-5);
        \draw[dashed,help lines] (G-3) to (G-4);
        \draw[dashed,help lines] (G-2) to (G-5);
        \draw (G-2) to (G-3);
        \draw (G-1) to (G-2);
        \draw[dashed,help lines] (G-0) to (G-5);
        \draw[dashed,help lines] (G-0) to (G-1);
        
    \end{tikzpicture}
    }
    \caption{$\mathcal{G}(1)$}%
    \label{subfig:t1}
\end{subfigure}
\begin{subfigure}{0.14\textwidth}
    \centering
    \scalebox{0.7}{%
    \begin{tikzpicture}[auto,line width=5pt]
        \draw (0,0) pic (G) {base};
        
        \draw (G-4) to (G-5);
        \draw[dashed,help lines] (G-3) to (G-4);
        \draw (G-2) to (G-5);
        \draw (G-2) to (G-3);
        \draw (G-1) to (G-2);
        \draw[dashed,help lines] (G-0) to (G-5);
        \draw[dashed,help lines] (G-0) to (G-1);
        
    \end{tikzpicture}
    }
    \caption{$\mathcal{G}(4)$}%
    \label{subfig:t4}
\end{subfigure}
\begin{subfigure}{0.14\textwidth}
    \centering
    \scalebox{0.7}{%
    \begin{tikzpicture}[auto,line width=5pt]
        \draw (0,0) pic (G) {base};
        
        \draw (G-4) to (G-5);
        \draw (G-3) to (G-4);
        \draw (G-2) to (G-5);
        \draw (G-2) to (G-3);
        \draw (G-1) to (G-2);
        \draw (G-0) to (G-5);
        \draw (G-0) to (G-1);
        
    \end{tikzpicture}
    }
    \caption{$\mathcal{G}(26)$}%
    \label{subfig:t26}
\end{subfigure}
\begin{subfigure}{0.14\textwidth}
    \centering
    \scalebox{0.7}{%
    \begin{tikzpicture}[auto,line width=5pt]
        \draw (0,0) pic (G) {base};
        
        \draw[dashed,help lines] (G-4) to (G-5);
        \draw[dashed,help lines] (G-3) to (G-4);
        \draw[dashed,help lines] (G-2) to (G-5);
        \draw[dashed,help lines] (G-2) to (G-3);
        \draw[dashed,help lines] (G-1) to (G-2);
        \draw[dashed,help lines] (G-0) to (G-5);
        \draw[dashed,help lines] (G-0) to (G-1);
        
    \end{tikzpicture}
    }
    \caption{$\mathcal{G}(33)$}%
    \label{subfig:t33}
\end{subfigure}
    \caption{Example EPG $\mathcal{G}$ and the snapshot graphs corresponding to
        $t \in \{0,1,4,\linebreak[4]26,45\}$. $\mathcal{G}$ has a period of length $60$.
        It illustrates the blow-up in complexity due to the compact representation.
        For example, the first $K_2$-free snapshot graph is at time step $33$.
    }%
    \label{fig:example}
\end{figure}

As humans tend to follow a daily routine and the systems that are to be described by time-varying graphs are mostly influenced by human behavior, they naturally exhibit a periodic behavior. 
For instance, in social networks describing the dynamics of people meeting~\cite{10.1371/journal.pone.0130824,leskovec2014snap}, the whole network will be quite complex, but every person individually follows mostly a daily routine. Hence, in order to describe the system compactly as an EPG we only need to consider the daily routine of two people at the same time to specify an edge. 
An other example is to model a train network. There, the underlying graph represents the railway system, while an edge is present in a time step if and only if a train is scheduled to run on the respective rail segment at that time.
A major advantage of modelling a time-varying system with EPGs is that, if for some application, we are only interested in a part of the temporal graph (for instance, we are only interested in the train schedule of a commune and not of the whole state), then we can first extract the corresponding subgraph of $\mathcal{G}$ and then compute the sequence of snapshot graphs, which will be both, smaller in the size of the individual snapshots, and the sequence might be shorter as the period of the sequence might be smaller. Hence, we avoid considering the complete huge and complicated system if we are only interested in a part of the system.

So far, to the best of our knowledge, the class of edge periodic (temporal) graphs is not studied in detail, yet. We counter this by giving a fundamental analysis of the parameterized complexity of essential graph-theoretical problems on EPGs such as being minor- or subgraph-free, containing a minor or subgraph, and the fundamental short traversal problem~\cite{DBLP:journals/pvldb/WuCHKLX14,DBLP:journals/jcss/AkridaMNRSZ20} from the theory of time-varying graphs.
The theory on graph minors, established by Robertson and Seymour in a series of over 20 publications~\cite{lovasz2006graph}, is one of the most fundamental results in graph theory. They showed that minor closed properties of a given graph can be checked in polynomial time 
as the minor relation is a well-quasi-ordering and hence, every minor closed family excludes a finite set of minimal minors. This implies that in order to recognize a minor closed family one only needs to 
test a finite number of minors and the latter task can be done in time $f(|H|)\cdot n^2$~\cite{DBLP:journals/jct/KawarabayashiKR12}, where $H$ is the sought minor. As the finite set of minors is fixed with respect to the graph property, these tests can be performed in polynomial time. Hence, it is natural to ask, if the toolbox of minors carries over to EPGs. 
For those, one could be interested in two questions: (1) Do all snapshot graphs obey a minor closed property? (2) Is there some snapshot graph that obeys a minor closed property? As those properties are proved by excluding certain minors, question (1) relates to a no-answer to the question whether there exists a snapshot graph containing a certain minor and question (2) relates to a yes-answer to the question whether there exists a snapshot graph being minor-free. Note that for EPGs, it could be that the underlying graph is not contained in a minor closed graph class but still each snapshot is contained.

While classically, both problems of being minor-free and finding a minor are \FPT in the size of the sought minor, we will observe that for EPGs, both problems are \NP-hard even if the minor is fixed and very simple, such as a triangle, or a a star with four leafs.
This implies that the graph minor toolbox does not translate to~EPGs.
In fact, our \NP-hardness results hold even in the case of topological minors.
On the other hand, the problem of finding a subgraph is not getting harder when we shift from classic graphs to EPGs. This problem is classically \W{1}-hard for the size of the subgraph (consider cliques as subgraphs)~\cite{DBLP:books/sp/CyganFKLMPPS15} and in \XP for the same parameter. Surprisingly, we can obtain a similar \XP-algorithm for EPGs in the same parameter.
For the problem of checking whether there is a snapshot graph that does not contain 
a fixed subgraph/minor, we obtain \NP-completeness for both problems, while if the 
  sought subgraph/minor is given in the input (and hence not fixed), we lift the \coNP-completeness from the classic setting to $\Sigma_2^P$-completeness concerning EPGs. 
Despite the high complexity, we present \FPT-algorithms in a combined parameter, including the size of the underlying graph, for all four problems of containment/freeness of minors/subgraphs.
We indicate that the parameter $|G|$ is necessary by giving hardness results when~$|G|$ is replaced by smaller structural parameters such as vertex cover number, treewidth, and pathwidth of the underlying graph. 

We emphasize that EPGs can trivially be converted into temporal graphs by unrolling the whole sequence of snapshot graphs in exponential time and space. Hence, the apparent complexity blow-up comes from the compact representation via periodic edge labels. Intuitively, as for encoding a problem in binary instead of unary, we do not need more time than for temporal graphs, we are just measuring in a smaller input size. But we
can exploit the additional structure of EPGs to obtain better algorithms than with the naive approach of unrolling the EPG.


%




\section{Preliminaries}
For a string $w= w_0w_1\dots w_n$ with $w_i \in \{0, 1\}$, for $0 \leq i \leq n$, we denote with $w[i]$ the symbol $w_i$ at position $i$ in $w$. Let $|w| = n$ be the \emph{length} of $w$.
We~write the concatenation of strings~$u$ and~$v$ as~$u \cdot v$. 
For non-negative integers $i \leq j$ we denote with $[i, j]$ the interval of natural numbers $n$ with $i \leq n \leq j$.
A \emph{monomorphism}~$\varphi \colon V \to V'$ is an isomorphism when restricted to its image. For a set $S = \{s_1, s_2, \dots, s_n\}$, we might denote the set $\{\varphi(s_1), \varphi(s_2), \dots, \varphi(s_n)\}$~by~$\varphi(S)$.

An \emph{edge periodic (temporal) graph}, \emph{EPG} for short, $\mathcal{G}=(V,E,\tau)$ (see also \cite{erlebach2020game})
consists of a graph $G=(V,E)$ (called the \emph{underlying graph})
and a function~$\tau:E \to \{0,1\}^*$ where $\tau$ maps each edge~$e$ to a string~$\tau(e)\in\{0,1\}^*$
such that $e$~exists in a time step~$t\geq 0$ if and only if $\modid{\tau(e)}{t} = 1$, where~$\modid{\tau(e)}{t} := \tau(e)[t \mod |\tau(e)|]$.
For an edge $e$ and non-negative integers $i \leq j$, we inductively define  $\modid{\tau(e)}{[i,j]}=\modid{\tau(e)}{i} \cdot \modid{\tau(e)}{[i+1,j]}$ and $\modid{\tau(e)}{[j,j]} = \modid{\tau(e)}{j}$.
%
Every edge~$e$ exists in at least one time step, that is, for each edge~$e$ there is some $t_e\in [0, |\tau(e)|-1]$ with $\tau(e)[t_e] = 1$. 
We might abbreviate $i$ repetitions of the same symbol $\sigma$ in $\tau(e)$ as $\sigma^i$.
We call $\#1_\text{max}$ the maximal number of ones appearing in an edge label $\tau(e)$ over all edges $e\in E$. Similarly, we call $\#0_\text{max}$ the maximal number of zeros appearing in some $\tau(e)$. 

Let $L_{\mathcal{G}} = \{|\tau(e)| \mid e \in E\}$ be the set of all edge periods of some edge periodic graph $\mathcal{G}=(V, E, \tau)$ and let $\LCM(L_{\mathcal{G}})$ be the \emph{least common multiple} of all periods in $L_{\mathcal{G}}$.
We denote with $\mathcal{G}(t)$ the subgraph of $G$ present in time step $t$. We do not assume that $\mathcal{G}$ is connected in any time step. 
If not stated otherwise, we assume an edge periodic graph to be undirected.

For an EPG $\mathcal{G}=(V,E,\tau)$ we define the layered directed graph $\mathcal{G}_\circlearrowleft = (V_\circlearrowleft, E_\circlearrowleft)$ where $V_\circlearrowleft = V \times [0, \LCM(L_\mathcal{G})-1]$ and two vertices $(u, i), (v, j) \in V_\circlearrowleft$ are connected in $E_\circlearrowleft$ with a directed edge $((u, i), (v, j))$, if $j = i+1 \mod \LCM(L_\mathcal{G})$ and either $u = v$, or $\{u, v\} \in E$ and $\modid{\tau(\{u, v\})}{i} = 1$. Intuitively, $\mathcal{G}_\circlearrowleft$ enrols the periodic temporal graphs and describes how we can traverse between the vertices taking the current time step into account.

\section{Periodic Character Alignment}
Most of our hardness results presented in this work will be based on the \PCA problem which was shown to be \NP-complete in~\cite{DBLP:conf/mfcs/MorawietzRW20}. This problem builds a bridge between the modern setting of edge periodic temporal graphs and the classical field of automata theory as it is closely related to the \textsc{Intersection Non-Emptiness} problem of deterministic finite automata over a unary alphabet.

\prob{\PCA (\textsc{PCA})}{A finite set~$X\subseteq \{0,1\}^*$ of binary strings.}{Is there a position~$i$, such that~$\modid{x}{i} = 1$ for all~$x\in X$?}

The parameterized complexity of \textsc{PCA} was already considered in~\cite{DBLP:conf/mfcs/MorawietzRW20} where \W1-hardness was shown for the parameter $|X|$ and \FPT-algorithms were given for the total number of runs of 1's, in all strings, the combined parameter~$|X|$ plus the greatest common divisor of any pair of lengths of strings of~$X$, and the length of the longest string in $X$. 
Here, a \emph{run} is a nonextendable (with the same minimal period) periodic segment in a string.
As the reductions from \textsc{PCA}, presented in this work, are parameter preserving, we inherit several \W1-hardness results from \textsc{PCA} for the different problems introduced for EPGs. 
Due to this tight connection, we begin with a more detailed analyzes of the parameterized complexity of the \textsc{PCA} problem.

\begin{restatable}{theorem}{restPCAconstnumberzeroes}\label{PCA const number zeroes}
\textsc{PCA} is \NP-hard even if $\#0_\text{max} = 1$.
\end{restatable}
\begin{toappendix}
\begin{proof}
	Let~$X$ be an instance of~\PCA.
	We describe how to obtain an equivalent instance~$X'$ of~\PCA in polynomial time such that each~$x' \in X'$ contains only a single 0.
	To obtain~$X'$, we start with an empty set and add for each~$x\in X$ and each~$i\in [0, |x| - 1]$ with~$x[i] = 0$, a string~$x_i$ to~$X'$ where~$|x_i| = |x|$ and~$x_i$ contains exactly one 0 at position~$i$.
	The equivalence between~$X$ and~$X'$ now follows directly from the fact that for each~$t \geq 0$ there is an~$x\in X$ with~$\modid{x}{t} = 0$ if and only if~$\modid{x_{(t\mod|x|)}}{t} = 0$.
\end{proof}
\end{toappendix}

\begin{restatable}{theorem}{restPCAconstnumberones}\label{PCA const number ones}
	\textsc{PCA} is \NP-hard even if $\#1_\text{max} \leq 9$.
\end{restatable}

\begin{restatable}{corollary}{restPCAdifferentprimes}\label{PCA different prime factors}
	\textsc{PCA} is \W1-hard with respect to the number of different prime numbers in the prime factorizations of the integers in $L_{\mathcal{G}}$.
\end{restatable}
\begin{toappendix}
	
	To prove~\Cref{PCA const number ones} 
	and~\Cref{PCA different prime factors} 
	we recall the construction of the reduction from~\MCC to~\PCA~\cite{DBLP:conf/mfcs/MorawietzRW20}.

	\prob{\MCC}{A graph~$G=(V,E)$, an integer~$k$, and a~$k$-partition~$(V_1, \dots, V_k)$ of~$G$.}{Is there a vertex~$v_i\in V_i$ for each~$i\in[1,k]$ such that~$\{v_i\mid i\in[1,k]\}$ is a clique in~$G$.}
	
	\paragraph{Constructing an equivalent instance of~\PCA.}
	Let~$I=(G=(V,E), k , (V_1, \dots, V_k))$ be an instance of~\MCC.
	We describe how to obtain an equivalent instance~$X(I)$ of~\PCA in polynomial time.
	For each~$i\in[1,k]$, we compute a prime number~$p_i$ with~$|V_i|\leq p_i$ such that~$p_i$ and~$p_j$ are distinct for~$j\neq i$.
	Computing such prime numbers can be done in polynomial time~\cite{rosser1941explicit}.
	Moreover, let~$v_i^0, \dots, v_i^{|V_i|-1}$ denote the vertices of~$V_i$ for each~$i\in[1,k]$.
	
	For each pair of distinct~$i,j\in[1,k]$ with~$i < j$ we define a string~$x_{i,j}$ that represents the edges between~$V_i$ and~$V_j$ in~$G$.
	The string~$x_{i,j}$ has length~$p_i \cdot p_j$ and we set $$x_{i,j}[t] := \begin{cases}
		0 & t_i \geq |V_i| \\
		0 & t_j \geq |V_j| \\
		1 & \{v_i^{t_i}, v_j^{t_j}\} \in E  \\
		0 & \{v_i^{t_i}, v_j^{t_j}\} \notin E 
	\end{cases}$$
	where~$t_i := t \mod p_i$ and~$t_j := t \mod p_j$.
	The instance~$X(I)$ of~\PCA is now defined as~$X(I) := \{x_{i,j}\mid 1 \leq i < j \leq k\}$.
	
	\begin{theorem}[Lemma~5 in \cite{DBLP:conf/mfcs/MorawietzRW20}]
		$I$ is a yes-instance of~\MCC if and only if~$X(I)$ is a yes-instance of~\PCA.
	\end{theorem}

	With the construction of~$X(I)$, we can now prove~\Cref{PCA const number ones} and~\Cref{PCA different prime factors}.
	
	\begin{proof}[Proof of~\Cref{PCA const number ones}]
		The classical reduction by Karp from~\textsc{3-SAT} to~\textsc{Clique} implies that~\MCC is~\NP-hard even if~$|V_i| = 3$ for each~$i\in[1,k]$.
		Let~$I=(G,k,(V_1, \dots, V_k))$ be an instance of~\MCC with~$|V_i| = 3$ for each~$i\in[1,k]$.
		Since the number on 1's in a string~$x_{i,j}\in X(I)$ is equal to the number of edges between~$V_i$ and~$V_j$ in~$G$, and~$|V_i| \cdot |V_j| \leq 9$,~\PCA is \NP-hard even if~$\#1_\text{max} \leq 9$.
	\end{proof}
	
	\begin{proof}[Proof of~\Cref{PCA different prime factors}]
		Since~\MCC is~\W1-hard when parameterized by~$k$~\cite{DBLP:series/mcs/DowneyF99} and the length of the string~$x_{i,j}$ is the product of the two prime numbers~$p_i$ and~$p_j$, the set of prime factors of the strings of~$X(I)$ is exactly the set~$\{p_i\mid 1 \leq i \leq k\}$.
		Hence,~\PCA is~\W1-hard when parameterized by the number of different prime factors of strings of~$X(I)$
	\end{proof}
\end{toappendix}

\begin{toappendix}
	\paragraph{Multicolored Variant of \textsc{PCA}}
For some reductions we will use a generalization of \textsc{PCA} which was considered in~\cite{DBLP:conf/mfcs/MorawietzRW20}. In this variant, not all strings of~$X$ have to align at a common 1 but at least~$k$, respecting some partition constraints. 

\prob{\textsc{Multicolored \PCA} (\MCPCA)}{Finite sets~$X_1, \dots, X_k\subseteq \{0,1\}^*$ of binary stings.}{Is there a position~$i$, such that for each~$j\in[1,k]$, there is some~$x_j\in X_j$ with~$\modid{x_j}{i} = 1$.}

It was shown that~\MCPCA is~\NP-hard and \W1-hard when parameterized by~$k$ even if every string contains only a single~1~\cite{DBLP:conf/mfcs/MorawietzRW20}.
\end{toappendix}

The core task of the problems introduced in the next sections is to determine whether a certain graph structure exists in one time step or over a sequence of consecutive time steps.
As the existence of an edge $e$ in an EPG is determined by a binary string $\tau(e)$, we associate each EPG with a corresponding \textsc{PCA} instance.  
Hence, if the location of the sought graph structure in the underlying graph of the EPG is known, the problem of finding a time step in which the structure exists is equivalent to finding a time step in which the 1's of the corresponding \textsc{PCA}-instance align.
\begin{definition} \label{def:embedding}
	Let~$X$ be an instance of PCA. A triple~$(\mathcal{G},H,\varphi)$, where~$\mathcal{G}=(V,E,\tau)$ is an EPG, $H=(V_H,E_H)$ is a subgraph of the underlying graph~$G=(V,E)$ of~$\mathcal{G}$, and $\varphi : V_H \rightarrow V$ is a monomorphism that identifies~$H$ in~$G$, is called an~\emph{$X$-embedding}~if $\tau(E_H) = X$.
\end{definition}

\begin{restatable}{lemma}{restlemembedding}\label{lem:embedding}
	Let~$X$ be an instance of PCA and let~$(\mathcal{G},H,\varphi)$ be an~$X$-embedding. Then, there exists a time step~$t$ in which~$\varphi(H)$ exists in~$\mathcal{G}(t)$ if and only if~$X$ is a yes-instance of PCA.
\end{restatable}
\begin{toappendix}
\begin{proof}
	Assume, there exists a time step $t$ in which $\varphi(H)$ exists in $\mathcal{G}(t)$. Then, by definition, for each edge $e \in E_H$ it holds that $\modid{\tau(\varphi(e))}{t} = 1$. As $(\mathcal{G}, H, \varphi)$ is an $X$-embedding, we have that $\tau(E_H) = X$. Hence, for each element $x \in X$, we find an edge $e \in E_H$, with $x = \tau(\varphi(e))$. Hence, we have that $\modid{x}{t} = \modid{\tau(\varphi(e))}{t} = 1$.
	
	For the other direction, assume there exists a times-step $t$ for which $\modid{x}{t} = 1$, for each $x \in X$. Then, as $(\mathcal{G}, H, \varphi)$ is an $X$-embedding, for each edge $e \in E_H$, we find an element $x\in X$ such that $\tau(\varphi(e)) = x$ and therefore, $\modid{x}{t} = \modid{\tau(\varphi(e))}{t} = 1$ and $\varphi(H)$ exists in $\mathcal{G}(t)$.
\end{proof}
\end{toappendix}

Due to the close relation of EPGs and PCA stated in Lemma~\ref{lem:embedding} we immediately inherit hardness results from \textsc{PCA} for the problem of finding a given subgraph in some $\mathcal{G}(t)$ if $t$ is unknown. On the other hand, we can use the known \FPT-algorithms for \textsc{PCA} as subroutines in \FPT-algorithms for problems concerning EPGs. For instance, in the problem of finding a time step $t$ in which some graph $H$ is a subgraph of $\mathcal{G}(t)$, we can 
iterate over all subgraphs containing $H$ that can appear as a snapshot graph 
and then use the algorithms for \textsc{PCA} to find a time step in which this snapshot graph exists. We refer to Theorem~\ref{thm:collectionFPT} for details.

We are now ready to shift our view to edge periodic temporal graphs.

\section{Short Traversal}
As edge periodic temporal graphs model periodic connectivity in a graph, the most natural question is to ask what is the shortest traversal time between two vertices $a$ and $b$, taking into account the periodicity of edges. In other words, we want to know the most favourable time step $t$ to start the traversal from $a$ in order to have the shortest traversal time. Stated as a decision problem we obtain the \ST problem.


\prob{\ST (\textsc{EPG-ST})}{Edge periodic graph $\mathcal{G}=(V,E,\tau)$, vertices $a, b\in V$, and $k\in \mathbb{N}$.}{Is there a time step~$t$ such that starting from vertex~$a$ at time step~$t$, we can reach vertex~$b$ at the beginning of time step~$t+k$ while traversing at most one edge per time step?}

\begin{toappendix}
For the next results, we introduce the notation of string shifts.
Let~$x$ be a binary string and let~$i$ be an integer.
We denote the~\emph{left shift}~$x^{\leftarrow i}$ as the binary string where~$|x^{\leftarrow i}| := |x|$ and~$x^{\leftarrow i}[j] := \modid{x}{j + i}$ for each~$j\in[0, |x|-1]$.
\linebreak[5]Analogously, we define the~\emph{right shift}~$x^{\rightarrow i}$ as the inverse operation.
Note that $\modid{x^{\leftarrow i}}{j} = \modid{x}{j + i}$ and~$\modid{x^{\rightarrow i}}{j} = \modid{x}{j - i}$.
\end{toappendix}

\begin{restatable}{theorem}{restShortTravGkpath}
	\ST is \NP-hard and \W1-hard with respect to the combined parameter $|G| + k$ even if~$G$ is a path.
\end{restatable}
\begin{toappendix}
Intuitively, the above result is obtained by a reduction from \textsc{PCA} where the strings of the PCA instance are put as labels on an $(a,b)$-path of length $k$ and the label of the $i$'th edge is shift $i$ positions to the right. 
Therefore, if the PCA instance aligns at a common $1$, the path is appearing edge by edge in the order of the path allowing for a traversal without any delay. 
\begin{proof}
We reduce from \textsc{PCA}.
Let~$X=\{x_1, \dots, x_{|X|}\}$ be an instance of \textsc{PCA}.
We describe how to obtain an equivalent instance~$I:=(\mathcal{G}=(V,E,\tau),a, b,k)$ of~\ST, where the underlying graph~$G$ has~$|X|$ edges and where~$k = |X|$.
The~\W1-hardness of~\ST when parameterized by~$|G| + k$ follows then directly from the~\W1-hardness of~\PCA when parameterized by~$|X|$.
We set~$V:=\{v_j\mid j\in[0,|X|]\}$ and $E:=\{e_j := \{v_{j-1}, v_j\}\mid j\in[1, |X|]\}$.
Hence,~$G$ is a path with~$|X|$ edges.
Moreover, we set~$\tau(e_i) := {x_i}^{\rightarrow (i-1)}$,~$a := v_0$,~$b := v_{|X|}$, and~$k := |X|$.

This completes the construction of~$I$.
Next, we show that~$X$ is a yes-instance of \textsc{PCA} if and only if~$I$ is a yes-instance of~\ST.

$(\Rightarrow)$
Let~$t$ be an index such that~$\modid{x}{t} = 1$ for each~$x\in X$.
We show that, starting at time step~$t$ at vertex~$a$, we can reach vertex~$b$ at time step~$t+|X|$ by only traversing one edge per time step.
By construction~$\tau(e_i) = {x_i}^{\rightarrow (i-1)}$ and thus~$\modid{\tau(e_i)}{t + i - 1} = \modid{x}{t} = 1$ for each~$i\in[1,|X|]$.
Hence, in time step~$t + i - 1$ the edge~$e_i$ can be traversed and thus, in time step~$t+k-1$ one can reach vertex~$b$ when starting from vertex~$a$ at time step~$t$.
Thus,~$I$ is a yes-instance of~\ST.

$(\Leftarrow)$
Suppose that~$I$ is a yes-instance of~\ST.
Since the unique~$(a,b)$-path in~$G$ contains~$k = |X|$ edges, there is a time step~$t$ such that~$\modid{\tau(e_i)}{t+i-1} = 1$ for each~$i\in [1,|X|]$.
By construction~$\modid{\tau(e_i)}{t+i-1} = \modid{{x_i}^{\rightarrow (i-1)}}{t+i-1} = \modid{x_i}{t}$ and thus~$\modid{x_i}{t} = 1$.
Hence,~$X$ is a yes-instance of \textsc{PCA}.
\end{proof}
\end{toappendix}

\begin{restatable}{theorem}{restShortTravk}
	\ST is \W1-hard when parameterized by the vertex cover number of the underlying graph and~$k$, even if $\#1_\text{max} = 1$. 
\end{restatable}
\begin{toappendix}
\begin{proof}
We reduce from~\MCPCA which is~\W1-hard when parameterized by~$k$ even if each string contains at most one 1.

Let~$I=(X_1, \dots, X_k)$ be an instance of~\MCPCA.
We describe how to obtain an equivalent instance~$I'=(\mathcal{G}=(V,E,\tau),a, b,k')$ of~\ST, where the underlying graph~$G$ has a vertex cover of size~$|X|+1$ and where~$k'= 2\cdot k$ and the label of each edge contains exactly one 1.
Let~$X_i := \{x_i^1, \dots, x_i^{|X_i|}\}$ for each~$i\in [1,k]$.
We start with an empty graph~$G$ and add vertices~$v_0, \dots, v_k$ to~$G$.
Afterwards, we add for each~$i\in[1,k]$ and each~$j\in[1,|X_i|]$ a vertex~$v_i^j$ to~$G$ which is adjacent to exactly the vertices~$v_{i-1}$ and~$v_i$.
This completes the construction of the underlying graph~$G$.
The edge labels are assigned as follows: for each~$i\in[1,k]$ and each~$j\in[1,|X_i|]$ we set~$\tau(\{v_{i-1}, v_i^j\}) := {x_i^j}^{\rightarrow 2\cdot (i-1)}$ and~$\tau(\{v_i^j, v_i\}) := {x_i^j}^{\rightarrow 2\cdot (i-1) + 1}$.
Finally, we set~$k' := 2k$,~$a := v_0$ and~$b:= v_k$.

The idea of the construction is, that starting from vertex~$v_{i-1}$ at time step~$t$, one can reach vertex~$v_{i}$ at the beginning of time step~$t+2$ if and only if there is some~$x_i \in X_i$ with~$\modid{x_i}{t - 2 (i-1)} = 1$.

Note that~$\{v_0, \dots, v_k\}$ is a vertex cover of size~$k+1$ for~$G$ and that any shortest~$(a,b)$-path in~$G$ has length~$2k$.
Next, we show that~$I$ is a yes-instance of~\MCPCA if and only if~$I'$ is a yes-instance of~\ST.

$(\Rightarrow)$
Let~$t\in \mathds{N}$ and let for each~$i\in [1,k]$,~$j_i \in [1, |X_i|]$ such that~$\modid{x_i^{j_i}}{t} = 1$.\linebreak[5]
We show that starting at time step~$t$ at vertex~$a = v_0$ one can follow the path~$P=(v_0, v_1^{j_1}, v_1, \dots, v_k^{j_k}, v_k)$ and reach vertex~$v_k = b$ in at most~$k'$ time steps.

By construction,~$\tau(\{v_{i-1}, v_i^{j_i}\}) = {x_i^j}^{\rightarrow 2(i-1)}$ and~$\tau(\{v_i^{j_i}, v_i\}) = {x_i^j}^{\rightarrow 2(i-1)+1}$ and thus~$\modid{\tau(\{v_{i-1}, v_i^{j_i}\})}{t + 2(i-1)} = \modid{\tau(\{v_i^{j_i}, v_i\})}{t + 2(i-1) + 1} = \modid{x_i^j}{t} = 1$ for each~$i\in[1,k]$.
Hence, starting from time step~$t$ at vertex~$v_0$, one can traverse each edge~$\{v_{i-1}, v_i^{j_i}\}$ at time step~$t+2(i-1)$ and each edge~$\{v_i^{j_i}, v_i\}$ at time step~$t+2(i-1)+1$ and reach vertex~$v_k$ within~$k'$ time steps.
As a consequence,~$I'$ is a yes-instance of~\ST.

$(\Leftarrow)$
Suppose that~$I'$ is a yes-instance of~\ST.
Since any~$(a,b)$-path in~$G$ contains at least~$k'$ edges, there is a time step~$t$, for each~$i\in[1,k]$ an index~$j_i\in [1, |X_i|]$ such that~$\modid{\tau(\{v_{i-1}, v_i^{j_i}\})}{t+2(i-1)} = \modid{\tau(\{v_i^{j_i}, v_i\})}{t+2(i-1)+1} = 1$ for each~$i\in[1,k]$.
By construction, we have $\tau(\{v_{i-1}, v_i^{j_i}\}) = {x_i^j}^{\rightarrow 2(i-1)}$ and~$\tau(\{v_i^{j_i}, v_i\}) = {x_i^j}^{\rightarrow 2(i-1)+1}$ for each~$i\in[1,k]$, and thus,~$\modid{x_i^{j_i}}{t} = 1$.
Hence,~$I$ is a yes-instance of~\MCPCA.
\end{proof}
\end{toappendix}

In contrast, if we combine the size of the underlying graph and the maximal number of ones per edge label, we can obtain an \FPT-algorithm. Note that the length of each edge label $\tau(e)$, and therefore $\LCM(L_\mathcal{G})$, is not restricted by the combination of parameters.
\begin{theorem}
	\ST is \FPT with respect to the combined parameter $|G| + \#1_\text{max}$ and can be solved in $\Oh(|G| \cdot \#1_\text{max})^{\Oh(|G| \cdot \#1_\text{max})} \cdot |\mathcal{G}|^{\Oh(1)}$~time.
\end{theorem}
\begin{proof}
Let~$I=(\mathcal{G}=(V,E,\tau),a, b,k)$ be an instance of~\ST.
To obtain an \FPT-algorithm, we perform two steps:
First, we iterate over all possible~$(a,b)$-paths~$P=(v_0, \dots, v_r)$ in the underlying graph~$G$, where~$v_0= a$ and~$v_r = b$.
Since we can assume that the temporal walk with the shortest traversal time is vertex simple, that is, each vertex is visited at most once, it remains to show that there is a time step~$t$ and an~$(a,b)$-path in the underlying graph, such that at time step~$t$ one can start at vertex~$a$ and reach vertex~$b$ in at most~$k$ time steps by only traversing edges of the path~$P$.
To check if such a time step exists for a given path~$P$, we present the following ILP-formulation.

For each edge~$e_i := \{v_{i-1}, v_i\}$, we use a variable~$t_i$ which is equal to the time step in which the considered temporal walk with shortest traversal time traverses edge~$e_i$.
Since at most one edge can be traversed at a time step, we need to ensure that~$t_i + 1 \leq t_{i+1}$.
Moreover, an edge~$e_i$ can only be traversed at time step~$t_i$, if~$\tau(e_i)[t_i\mod |\tau(e_i)|] = 1$.
Hence, we first introduce two additional variables~$c_i$ and~$m_i$ for each edge~$e_i$, where~$m_i \in [0, |\tau(e_i)|-1]$ and~$|\tau(e_i)|\cdot c_i + m_i = t_i$.
That is,~$m_i$ stores the value of~$t_i\mod |\tau(e_i)|$.
Finally, we have to ensure that~$\tau(e_i)[m_i] = 1$.
Let~$J_i := \{j\in [0, |\tau(e_i)|]\mid \tau(e_i)[j] = 1\}$ denote the set of positions where~$\tau(e_i)$ is equal to one.
We introduce for each~$i\in [1,r]$ and each~$j\in J_i$ a new binary variable~$\ell_{i,j}\in\{0,1\}$ which is equal to zero if and only if~$m_i = j$.
To make sure that~$\tau(e_i)[m_i] = 1$, the value of exactly one~$\ell_{i,j}$ has to be zero, which can be achieve by adding the constraint~$\sum_{j\in J_i} \ell_{i,j} = |J_i|-1$. 
The complete ILP formulation now reads as follows: 
\begin{align}
  t_{i}, c_i \in \mathds{N} &&& \text{ for each $i\in[1,r]$}\nonumber\\
  m_i \in \{0,|\tau(e_i)|-1\} &&& \text{ for each $i\in[1,r]$}\nonumber\\
  \ell_{i,j} \in \{0,1\} &&& \text{ for each $i\in[1,r], j\in J_i$}\nonumber\\
\intertext{Minimize $t_r - t_1$ subject to}
  t_i + 1 \leq t_{i+1} &&& \text{ for each $i\in[1,r-1]$}\nonumber \\
  c_i \cdot |\tau(e_i)| + m_i = t_i &&& \text{ for each $i\in[1,r]$}\nonumber\\
  - \ell_{i,j} \cdot 2 |\tau(e_i)| + j \leq m_i &&& \text{ for each $i\in[1,r], j\in J_i$}\nonumber\\
  \ell_{i,j} \cdot 2 |\tau(e_i)| + j \geq m_i &&& \text{ for each $i\in[1,r], j\in J_i$}\nonumber\\
  \sum_{j\in J_i} \ell_{i,j} = |J_i| - 1 &&& \text{ for each $i\in[1,r]$}\nonumber
\end{align}

\noindent Note that the number of variables in this ILP-formulation is~$\Oh(r \cdot \max_{i\in [1,r]}|J_i|)$.
Since~$r \leq |G|$ and~$\max_{i\in [1,r]}|J_i| \leq \#1_\text{max}$, this ILP can be solved in $\Oh(|G| \cdot \#1_\text{max})^{\Oh(|G| \cdot \#1_\text{max})} \cdot |I|^{\Oh(1)}$~time~\cite{DBLP:books/sp/CyganFKLMPPS15}.

Since there are at most~$2^{|G|}$ many possible~$(a,b)$-paths in~$G$ and we can solve for each such path the corresponding ILP in $\Oh(|G| \cdot \#1_\text{max})^{\Oh(|G| \cdot \#1_\text{max})}  |I|^{\Oh(1)}$ time, \ST can be solved in the stated running time.
\end{proof}

\begin{remark}
Note that since any edge can be traversed within at most~$\max(L_\mathcal{G})$ time steps and any vertex simple path contains at most~$n-1$ edges, any shortest temporal path from~$a$ to~$b$ requires at most~$\max(L_\mathcal{G}) \cdot (n-1)$ time steps.
Hence, any shortest path generalization on EPGs where we fix the start or the end time step can be solved in polynomial time, since we can simply reduce it back to the shortest path problems on at most~$\Oh(\max(L_\mathcal{G})\cdot n)$ consecutive layers of~$\mathcal{G}_\circlearrowleft$.
Examples for this would be a~\textsc{EPG Shortest Arrival}, where we want to reach vertex~$b$ as fast as possible or~\textsc{EPG Latest Departure}, where we want to find the latest time step~$t_0$ such that we can reach vertex~$b$ at the latest at time step~$t$ when starting from vertex~$a$ at time step~$t_0$.   
\end{remark}
\section{Minors and Subgraphs}
We now come to the main part of this paper considering the existence and non existence of sub-structures in an EPG such as induced subgraphs and minors. Recall that $G'=(V', E')$ is a \emph{subgraph} of a graph $G=(V, E)$ if $V' \subseteq V$ and $E' \subseteq E$. If further for all $u, v \in V'$ it holds that $\{u, v\} \in E'$ if and only if $\{u, v\} \in E$, we call $G'$ an \emph{induced subgraph} of $G$. In the following, we see subgraphs as induced subgraphs.
We call $G'$ a \emph{minor} of $G$ if $G'$ can be obtained from $G$, by deletion of vertices, deletion of edges, and contraction of edges.
Here, we consider the following questions: Does there exists a time step $t$, such that $\mathcal{G}(t)$ has an subgraph/minor or is subgraph-/minor-free.

\subsection{Subgraphs}

Now, we study the following two problems.

\noindent
\begin{minipage}{0.54\textwidth}
	\prob{\SG}{EPG $\mathcal{G}=(V,E,\tau)$ and graph $H=(V_H, E_H)$.}{Is there a time step $t$, s.t.\ $H$ is a subgraph of $\mathcal{G}(t)$?}
\end{minipage}
\begin{minipage}{0.54\textwidth}
	\prob{\SGfree}{EPG $\mathcal{G}=(V,E,\tau)$ and graph $H=(V_H, E_H)$.}{Is there a time step $t$, s.t.\ $H$ is \emph{not} a subgraph of~$\mathcal{G}(t)$?}
\end{minipage}

\begin{restatable}{theorem}{restSGNPh} \label{thm:SGNP-h}
	The \SG problem is \NP-complete and \W{1}-hard parameterized by $|G|$. This holds even if~$H$ is a path and~$G=H$.
\end{restatable}
\begin{proof}
\SG belongs to \NP, since we may non-deterministically choose a time step~$t$ of size at most $\LCM(L_{\mathcal{G}})$ and an embedding~$\varphi: V_H \rightarrow V$ and check, whether~$\varphi$ identifies~$H$ in~$\mathcal{G}(t)$. Since~$t \leq \max(L_{\mathcal{G}})^{(n^2)}$, this certificate can be encoded polynomially in the input~size.

It remains to show that \SG is \NP-hard. Let~$X:=\{x_1, \dots, x_{|X|}\}$ be an instance of PCA. We define an equivalent instance~$(\mathcal{G},H)$ of \SG. First, we define~$H:=(V_H,E_H)$ to be a path on~$|X|$ edges~$e_1, \dots, e_{|X|}$. Second, we define~$\mathcal{G}:= (V_H,E_H, \tau)$ with~$\tau(e_i):=x_i$ for every~$i\in [1,|X|]$.

We next use Lemma~\ref{lem:embedding} to show that~$X$ is a yes-instance of PCA if and only if~$(\mathcal{G},H)$ is a yes-instance of \SG. Observe that~$\varphi: V_H \rightarrow V_H$ with~$\varphi(v):=v$ is a trivial monomorphism that identifies~$H$ in the underlying graph of~$\mathcal{G}$. Furthermore, by the definition of~$\tau$ we have~$\tau(E_H)=X$. Thus,~$(\mathcal{G},H,\varphi)$ is an~$X$-embedding according to Definition~\ref{def:embedding}. Then, by Lemma~\ref{lem:embedding} we have that~$X$ is a yes-instance of PCA if and only if there is a time step~$t$ in which~$\varphi(H)$ exists in~$\mathcal{G}(t)$. Consequently,~$X$ is a yes-instance of PCA if and only if~$(\mathcal{G},H)$ is a yes-instance of \SG. 
\end{proof}

Note that the length of the paths in the construction behind~Theorem~\ref{thm:SGNP-h} corresponds to the size of the PCA instance. Thus, these paths might be arbitrarily long. If we---in contrast---assume that the size of sought subgraph~$H$ is bounded by some constant~$h$, we obtain a polynomial time algorithm for \SG. In other words, \SG is XP when parameterized by~$h$ as we show in the following theorem.

\begin{theorem}  \label{thm:SGXPh}
	\SG can be solved in time~$\Oh(n^{h} \cdot \max(L_\mathcal{G})^{(h^2)}) \cdot 2^{\Oh(\sqrt{h \log h})}$, where~$h$ is the number of vertices in~$H$. 
\end{theorem}
\begin{proof}
We prove the theorem by describing the algorithm. Let~$(\mathcal{G}=(V,E,\tau),H)$ be an instance of \SG. The algorithm is straight forward: We iterate over all possible subsets~$W \subseteq V$ of size~$h$. For each of these sets we check whether there is a time step~$t \in [1, \max(L_{\mathcal{G}})^{(h^2)}]$ such that~$\mathcal{G}(t)[W]$ is isomorphic to~$H$. If such a time step exists, return~\emph{yes}. Otherwise, return \emph{no}. 

The algorithm runs within the claimed running time since there are~$\binom{n}{h} \in \mathcal{O}(n^h)$ possible choices of~$W$. For each choice, we consider~$\max(L_\mathcal{G})^{(h^2)}$ distinct graphs~$\mathcal{G}(t)$ and check whether one of these graphs is isomorphic to~$H$ in~$2^{\Oh(\sqrt{h \log h})}$~time~\cite{BL83}.

We next show that the algorithm is correct. Suppose that the algorithm returns \emph{yes}. Then, for one choice of~$W$ and one time step~$t$, the graph~$\mathcal{G}(t)[W]$ is isomorphic to~$H$ and therefore,~$(\mathcal{G},H)$ is a yes-instance.

Conversely, suppose that~$(\mathcal{G},H)$ is a yes-instance. Let~$W \subseteq V$ be the subset of size~$h$ such that~$\mathcal{G}(t)[W]$ is isomorphic to~$H$ at some time step~$t$. Let~$e_1, \dots, e_k \in E$ be all edges between vertices of~$W$ in~$(V,E)$. Since~$|W|=h$ we have~$k \leq h^2$. Thus, the least common multiple of all string lengths $|\tau(e_1)| , \dots, |\tau(e_k)|$ is at most~$\max(L_{\mathcal{G}})^{(h^2)}$. Therefore, we may assume that~$t \in [1, \max(L_{\mathcal{G}})^{(h^2)}]$. Consequently, the algorithm returns \emph{yes}. 
\end{proof}

Next, we consider the problem \SGfree. Recall that Theorem~\ref{thm:SGXPh} reveals that the NP-hardness of \SG crucially relies on the fact that the size of~$H$ is unbounded. In contrast, we show next that \SGfree is NP-hard for every fixed size of~$H$.

\begin{theorem}
	\SGfree is \NP-complete and~\W{1}-hard parameterized by $|G|$ for every fixed subgraph~$H$ containing at least two vertices.
	\label{thm-subgrpah-free-np-h}
\end{theorem}

The containment stated in Theorem~\ref{thm-subgrpah-free-np-h} is easy to see: We non-deterministically choose a time step~$t$ of size at most the least common multiple of the individual edge label lengths. Note that~$t \leq \max(L_{\mathcal{G}})^{(n^2)}$ and thus,~$t$ can be encoded polynomially in the total input size. With~$t$ at hand, we check for every set~$V' \subseteq V$ of size~$h$, whether~$\mathcal{G}(t)[V']$ is not isomorphic to~$H$. Herein,~$h$ denotes the number of vertices of~$H$. Since~$H$ is a fixed subgraph,~$h$ is a constant and therefore, this can be done in polynomial time.

We next show the \NP-hardness from Theorem~\ref{thm-subgrpah-free-np-h} in two steps.
First, we provide NP-hardness for edgeless graphs~$H$ containing at least two vertices and second, we show NP-hardness for graphs~$H$ containing at least one edge.

\begin{lemma}
\SGfree is \NP-hard and~\W{1}-hard parameterized by $|G|$ for every fixed edgeless graph~$H$ containing at least two vertices.
\label{lem-subgrpah-free-np-h-1}
\end{lemma}

\begin{proof}
Let~$H$ be an edgeless graph on~$c \geq 2$ vertices. Note that~$H$ is a subgraph of some~$G$ if and only if~$G$ has an independent set of size~$c$. We prove the NP-hardness by providing a reduction from~PCA. Let~$X$ be an instance of PCA. Without loss of generality, we may assume that~$1 \in X$, since otherwise, we may replace~$X$ by the equivalent instance~$\widetilde{X}:=X \cup \{1\}$. Thus, let~$X:=\{x_1, \dots, x_k,1\}$.

We first describe the construction. We define an equivalent instance~$(\mathcal{G},H)$ of \SGfree. To this end, we define the auxiliary graph~$F:=(V_{F},E_{F})$ as a clique on~$2k$ vertices, and we let~$M := \{e_1, \dots, e_k\} \subseteq E_{F}$ be a matching of size~$k$ in~$F$. To define~$\mathcal{G}=(V,E,\tau)$, we let the underlying graph~$G=(V,E)$ be the disjoint union of~$F$ and~$c-2$ isolated vertices. We then define~$\tau$ by setting~$\tau(e_i):=x_i$ for every~$e_i \in M$ and~$\tau(e):=1$ for every~$e \not \in M$.

We next show that~$X$ is a yes-instance of~PCA if and only if~$(\mathcal{G},H)$ is a yes-instance of~\SGfree by applying Lemma~\ref{lem:embedding}. Let~$\varphi: V_F \rightarrow V$ be the monomorphism identifying~$F$ in~$G$. By construction of~$\tau$ we have~$\tau(E_F)=X$ and therefore,~$(\mathcal{G},F,\varphi)$ is an~$X$-embedding according to Definition~\ref{def:embedding}. Lemma~\ref{lem:embedding} implies that~$X$ is a yes-instance of PCA if and only if there is a time step~$t$ in which~$\varphi(F)$ exists in~$\mathcal{G}(t)$. It remains to show that~$\varphi(F)$ exists in~$\mathcal{G}(t)$ if and only if~$\mathcal{G}(t)$ does not have an independent set of size~$c$.

Suppose~$\varphi(F)$ exists in~$\mathcal{G}(t)$. Then,~$\mathcal{G}(t)$ is the union of~$c-2$ isolated vertices and a clique of size~$2k$. Thus, the maximum independent set in~$\mathcal{G}(t)$ has size~$c-1 < c$. Conversely, suppose~$\varphi(F)$ does not exist in~$\mathcal{G}(t)$. Then, one edge~$\{u,v\}$ with~$u \in V_F$ and~$v \in V_F$ is not present in~$\mathcal{G}(t)$. Thus, the vertices~$u$ and~$v$ together with~$c-2$ isolated vertices form an independent set of size~$c$ in~$\mathcal{G}(t)$. 
\end{proof}

\begin{restatable}{lemma}{restSubgraphFreeNP}
\SGfree is \NP-hard and~\W{1}-hard parameterized by $|G|$ for every fixed graph~$H$ containing at least one edge.
\label{lem-subgrpah-free-np-h-2}
\end{restatable}
\begin{toappendix}
\begin{proof}
Let~$H$ be a graph with at least one edge. Again, we provide a reduction from PCA. Let~$X=\{x_1, \dots, x_k\}$ be an instance of PCA.

We first describe the construction. We define an equivalent instance~$(\mathcal{G}=(V,E,\tau),H)$ of \SGfree. We set~$(V,E)$ to be the disjoint union of~$k$ copies of~$H_1, \dots, H_k$ of the graph~$H$. Furthermore, for every~$i \in [1,k]$ we set~$\tau(e):=\overline{x_i}$ for each edge~$e$ that belongs to the copy~$H_i$. Recall that for a string~$s$, the string~$\overline{s}$ is obtained by swapping every occurrence of a~$0$ with an occurrence of a~$1$ and vice versa.

We next show that~$X$ is a yes-instance of PCA if and only if~$(\mathcal{G},H)$ is a yes-instance of \SGfree.

Suppose that~$X$ is a yes-instance of PCA. Then, there is a position~$t$ such that~$x_i[t]^\circ=1$ for all~$x_i \in X$. Consequently, we have~$\overline{x_i}[t]^\circ=0$ for all~$x_i \in X$. Due to the construction of~$\tau$, this implies that~$\mathcal{G}(t)$ is an edgeless graph and therefore, $\mathcal{G}(t)$ does not contain~$H$ as an induced subgraph. Consequently,~$(\mathcal{G},H)$ is a yes-instance of \SGfree.

Conversely, suppose that~$X$ is a no-instance of PCA. Then, for every position~$t$ there is at least one~$x_i \in X$ with~$x_i[t]^\circ=0$. Consequently, for every~$t$ we have~$\overline{x_i}[t]^\circ=1$ for some~$x_i \in X$. Due to the construction of~$\tau$, this implies that at each time step~$t$, all edges of one of the copies~$H_i$ of~$H$ are present in~$\mathcal{G}(t)$. Therefore, every~$\mathcal{G}(t)$ contains~$H$ as an induced subgraph and therefore,~$(\mathcal{G},H)$ is a no-instance of \SGfree. 
\end{proof}
\end{toappendix}

Now, Theorem~\ref{thm-subgrpah-free-np-h} follows from~\Cref{lem-subgrpah-free-np-h-1} and~\Cref{lem-subgrpah-free-np-h-2}. In contrast, if the subgraph~$H$ is not fixed, then the problem becomes even harder. Intuitively, the following theorem is based on a construction from $\exists\forall$\textsc{3UNSAT}, where in the resulting EPG, we first have to guess a time step $t$ and then need to check that each selection of $k$ vertices is not a clique in~$\mathcal{G}(t)$.
\begin{restatable}{theorem}{restSubgraphFreeSigma}
	\label{them:SubgraphFreeSigma2}
	The \SGfree problem is $\Sigma_P^2$-complete.
\end{restatable}
\begin{toappendix}
\begin{proof}
	For the membership in $\Sigma_2^P$, we construct an alternating Turing machine~$M$ that solves the problem as follows. On input $\mathcal{G} = (V, E, \tau), H = (V_H, E_H)$, first existentially guess a time step $t$ for which the snapshot graph $\mathcal{G}(t)$ should be $H$-free. Then, universally guess a set of vertices $S$ with $|S| = |H|$ and a monomorphism $\varphi \colon V_H \to S$. If for all $\{u, v\} \in E_H$ it holds that $\varphi(\{u, v\}) \in E$, then return \textit{no}, else return \textit{yes}. Clearly, if $M$ outputs \textit{yes}, it found a snapshot graph $\mathcal{G}(t)$ not containing $H$ as a subgraph. As $M$ performs on every input an existential guess followed by a universal guess, this classifies the \SGfree problem to be contained in $\Sigma_2^P$.
	
	Next, we show that $\SGfree$ is $\Sigma_2^P$-hard. Therefore, we reduce from the complement of the $\Pi_2^P$-complete problem $\forall\exists$\textsc{3SAT}~\cite{schaefer2002completeness,DBLP:journals/tcs/Stockmeyer76,DBLP:journals/tcs/Wrathall76} which can be described as $\exists\forall$\textsc{3UNSAT}. 
	In $\exists\forall$\textsc{3UNSAT} we are given a quantified Boolean formula of the form $\exists \vec{x} \forall \vec{y} \colon \psi(\vec{x}, \vec{y})$ where $\vec{x}$ and $\vec{y}$ are existential, respectively universal, quantified variables and $\psi(\vec{x},\vec{y})$ is a Boolean formula without free variables. Then, the question is whether there exists an assignment for the variables in $\vec{x}$ such that for all assignments of the variables in $\vec{y}$, the formula $\psi(\vec{x}, \vec{y})$ evaluates to false.
	Intuitively, we will iterate through all possible assignments to $\vec{x}$ with the different time steps of an EPG and then use the reduction from \textsc{3SAT} to \textsc{Clique} presented in~\cite{CliqueReduction} for the universal quantified variables in $\vec{y}$. Then, a snapshot graph associated with an assignment for $\vec{x}$ will contain a clique of some size as a subgraph if and only if there exists an assignment for $\vec{y}$ that satisfies the formula $\psi(\vec{x},\vec{y})$.
	
	Let $\psi(\vec{x}, \vec{y})$ be a Boolean formula in 3CNF containing $k$ clauses. We construct a graph $G$ consisting of $k$ clusters with a maximum number of three vertices in each cluster. Each cluster will correspond to a clause of $\psi(\vec{x}, \vec{y})$. Each vertex in a cluster is assigned with a literal from the respective clause. Then, we put edges between all pairs of vertices from different clusters except for pairs of the form $x, \neg x$, being associated with a variable and its negation. There are no edges between vertices of the same cluster. Intuitively, if two vertices in $G$ are adjacent, their respective literals can be assigned true simultaneously. 
	Now, we still need to assign edge labels to the edges in $G$ which will become the underlying graph of the constructed EPG $\mathcal{G}$. 
	As an intermediate step, we assign labels to the endpoints of edges and obtain the label $\tau(e)$ for the edge $e$ by multiplying the labels $w_u$ and $w_v$ of the two endpoints of $e$ as follows. 
	If $|w_u| = \ell_u$ and $|w_v| = \ell_v$, then the label of $e$ will be of length $\ell = \LCM(\{\ell_u, \ell_v\})$ and is defined for $0 \leq i \leq l$ as $\modid{\tau(e)}{i} = \modid{w_u}{i} \wedge \modid{w_v}{i}$ where a position in a binary string is interpreted as a truth-value.
	We assign for an edge $e$, an endpoint incident to a literal corresponding to a universally quantified variable with the edge label $\tau(e) = 1$.
	Let $m$ be the number of existential quantified variables. Then, let $P = \{p_1, p_2, \dots, p_m\}$ be the set of the first $m$ prime numbers. Note that the $i$'th prime number is bounded by $\mathcal{O}(i\log i)$~\cite{rosser1941explicit}, hence, $P$ be computed in polynomial time.
	Then, for each remaining endpoint of an edge $e$ being incident to some literal associated with an existential quantified variable $x_i$, we assign a label $10^{p_i-1}$ to the endpoint of $e$ if $x_i$ appears as a positive literal and a label $01^{p_i-1}$, otherwise.
	A time step $t$ with $t\mod p_i =0$ will then correspond to setting variable $x_i$ to true and a time step $t$ with $t\mod p_i \neq 0$ will correspond to setting $x_i$ to false.
	Finally, we set the sought subgraph $H$ to a clique of size $k$.
	
	We claim that $\mathcal{G}$ has a time step $t$ in which $\mathcal{G}(t)$ is $H$-free if and only if $\exists \vec{x}\forall \vec{y}\psi(\vec{x}, \vec{y})$ evaluates to false.
	First, assume there exists a time step $t$ for which $\mathcal{G}(t)$ is $H$-free. 
	In $\mathcal{G}(t)$ 
	all edges incident to literals are present where both literals satisfy that they correspond either to literals of universally quantified variables, 
	to positive literals of existentially quantified variables $x_i$ with $t \mod i = 0$, or to negative literals of existentially quantified variables $x_i$ with $t \mod i \neq 0$.
	If one literal incident to an edge $e$ does not satisfy any of these conditions, then the edge $e$ is not present in $\mathcal{G}(t)$.
	Now, if $\mathcal{G}(t)$ has $k$-clique, it contains exactly one literal from each cluster. As all nodes of a clique are adjacent, all corresponding literals must be assigned true simultaneously. As by assumption, $\mathcal{G}(t)$ is $H$-free, per cluster, for each choice of the literal going into the clique, there is at least one edge missing between the chosen literals. This means that we picked two literals corresponding to positive and negative literals of the same variable. As picking a literal vertex from a cluster to go into the clique corresponds to choosing which literal satisfies the clause corresponding to the cluster, this means that we cannot make a selection of satisfying literals per clause that leads to a valid variable assignment. Hence, for all assignments of $\vec{y}$, the formula $\psi(\vec{x},\vec{y})$ evaluates to false, when $\vec{x}$ is assigned such that exactly those variables $x_i$ with $t \mod i = 0$ are assigned true. 
	
	For the other direction, assume there exists a variable assignment for $\vec{x}$ such that for all assignments of $\vec{y}$, the formula $\psi(\vec{x}, \vec{y})$ evaluates to false. Let $t$ be a time step such that for all variables $x_i$ assigned true, it holds that $t \mod i = 0$ and for variables $x_i$ assigned false, it holds that $t\mod i \neq 0$. As the edge labels for variables in $\vec{x}$ have different prime lengths we can find such a $t$ for every possible assignment. Now pick an assignment for the variables in $\vec{y}$. 
	We mark each literal that is satisfied under this assignment in $\mathcal{G}(t)$.
	By assumption, $\psi(\vec{x}, \vec{y})$ evaluates to false under the current assignment. As we are considering a full assignment of the variables, this means that there is one cluster that does not contain a marked literal vertex. Hence, regardless of which vertex $v$ of this cluster we would pick, we would find a marked literal of some other cluster which correspond to the complementary literal of $v$ and hence is missing an edge with~$v$. Therefore, we cannot complete the current vertex marking to a $k$-clique. As the assignment for~$\vec{y}$ was arbitrary it follows that for each possible embedding of $H$ into $\mathcal{G}(t)$ we would have an edge missing.
\end{proof}
\end{toappendix}

\subsection{Minors}

Now, we study the following two problems.\\

\noindent
\begin{minipage}{0.54\textwidth}
	\prob{\Minor}{EPG $\mathcal{G}=(V,E,\tau)$ and graph $H=(V_H, E_H)$.}{Is there a time step $t$, s.t.\ $H$ is a minor of $\mathcal{G}(t)$?}
\end{minipage}
\begin{minipage}{0.54\textwidth}
	\prob{\Minorfree}{EPG $\mathcal{G}=(V,E,\tau)$ and graph $H=(V_H, E_H)$.}{Is there a time step $t$, s.t.\ $H$ is \emph{not} a minor of~$\mathcal{G}(t)$?}
\end{minipage}\\

\noindent As in the subgraph variant, we obtain $\Sigma_2^P$-completeness for \Minorfree.

\begin{restatable}{theorem}{restMinorSigma}
	The \Minorfree problem is $\Sigma_2^P$-complete.
\end{restatable}
\begin{toappendix}
\begin{proof}
	This theorem follows easily from the proof of Theorem~\ref{them:SubgraphFreeSigma2} by changing the sought subgraph $H$, being a $k$-clique, to $H' = H \cup S$, where $S$ is an independent set of size $2k$.
	Now, the number of vertices in $H'$ is equal to the number of vertices in $G$. Hence, in order to obtain $H'$ as a minor, we cannot use edge contraction as we would thereby loose a vertex. As $H$ is a full $k$-clique, we can also not use edge or vertex deletion to obtain $H'$ and hence, $H$ must already appear as a subgraph in a considered snapshot graph. 
\end{proof}
\end{toappendix}

If we fix the considered minor, the complexity falls to \NP-completeness, which is still significantly harder than the polynomial time solvability in the case of classic static graphs~\cite{RobertsonS95b}.

\begin{restatable}{theorem}{restMinorFreeNP}
\Minorfree is \NP-complete and~\W{1}-hard parameterized by $|G|$ for every fixed~$H$ containing at least one edge.
	\label{thm-minor-free-np-h}
\end{restatable}
\begin{toappendix}
\begin{proof}
The proof is similar to the proof of Lemma~\ref{lem-subgrpah-free-np-h-2}:
We provide a reduction from PCA. 
Let~$X$ be an instance of PCA. 
Without loss of generality, we may assume that~$1 \in X$, since otherwise, we may replace~$X$ by the equivalent instance~$\widetilde{X}:=X \cup \{1\}$. 
Thus, let~$X:=\{x_1, \dots, x_k,1\}$.
We define an equivalent instance~$(\mathcal{G}=(V,E,\tau),H)$ of \Minorfree. 
We set~$(V,E)$ to be the disjoint union of~$k$ copies~$H_1, \dots, H_k$ of the graph~$H$. 
Furthermore, for every~$i \in [1,k]$ we set~$\tau(e):=\overline{x_i}$ for each edge~$e$ that belongs to the copy~$H_i$. 
Recall that for a string~$s$, the string~$\overline{s}$ is obtained by swapping every occurrence of a~$0$ with an occurrence of a~$1$ and vice versa.

The correctness now follows by the observation that~$\mathcal{G}(t)$ is~$H$-minor-free if and only if~$\mathcal{G}(t)$ does not contain~$H$ as a subgraph and the correctness of the reduction in Lemma~\ref{lem-subgrpah-free-np-h-2}. 
\end{proof}
\end{toappendix}

Next, we consider the case that~$H$ is edgeless.
As we can delete edges to obtain the sought minor, in contrast to the subgraph variant, 
we only need to compare the number of vertices in $H$ and $\mathcal{G}$.

\begin{restatable}{proposition}{restMinorFreePoly}
\Minorfree can be solved in linear time using logarithmic space for every fixed edgeless graph~$H$.
	\label{thm-minor-free-poly-edgeless}
\end{restatable}
\begin{toappendix}
\begin{proof}
Since by the definition of minors, edge deletions are a valid operation and since~$H$ is edgeless, it is sufficient to count the number of vertices in the underlying graph~$G$.
More precisely,~$\mathcal{G}(t)$ is~$H$-minor-free in every time step~$t$ if and only if~$|V(G)|<|V(H)|$.
Counting the number of vertices of the underlying graph can be done in linear time using only logarithmic space. 
\end{proof}
\end{toappendix}

Next, we study the related problem \Minor, in which we ask whether a graph~$H$ exists as a minor in some time step~$t$ in an EPG.
For finding an $H$-minor, the problem is already \NP-complete for very simple minors.
More precisely, we provide a dichotomy for minors of constant sizes into cases which are NP-complete and those which are solvable in polynomial time.
First, we provide NP-completeness for the case that~$H$ contains at least one cycle.

\begin{restatable}{theorem}{restMinorTriangle}
\Minor is \NP-complete and~\W{1}-hard parameterized by $|G|$ for every fixed~$H$ containing a cycle.
\label{thm-minor-containment-hard-triangle}
\end{restatable}
\begin{toappendix}
\begin{proof}
We present a reduction from PCA.
Let~$X$ be an instance of PCA. 
Without loss of generality, we may assume that~$1 \in X$, since otherwise, we may replace~$X$ by the equivalent instance~$\widetilde{X}:=X \cup \{1\}$. 
Thus, let~$X:=\{x_1, \dots, x_k,1\}$.

 We define an equivalent instance~$(\mathcal{G},H)$ of \Minor{} as follows: 
Let~$H$ be a graph containing a cycle.
Furthermore, let~$C$ be an induced cycle and let~$e$ be an arbitrary but fixed edge of~$C$.
Next, we construct the underlying graph~$G$ of~$\mathcal{G}$.
The graph~$G$ consists of a copy of~$H$ in which the edge~$e$ is subdivided exactly~$k-1$ times, that is, we replace~$e=\{v_1,v_{k+1}\}$ by the path~$v_1,v_2,\ldots, v_{k},v_{k+1}$ of~$k-1$ new intermediate vertices.
Observe that the total number of cycles in~$G$ is equal to the total number of cycles in~$H$.
It remains to define~$\tau$.
We set~$\tau(\{v_i,v_{i+1}\}):=x_i$ for every~$i\in [1,k]$.
For every remaining edge~$e'$ of~$E(G)$ we set~$\tau(e'):=1$.

We next use Lemma~\ref{lem:embedding} to show that~$X$ is a yes-instance of PCA if and only if~$(\mathcal{G},H)$ is a yes-instance of \Minor. 
Observe that~$\varphi: V(G) \rightarrow V(G)$ with~$\varphi(v):=v$ is a trivial monomorphism that identifies~$G$ with itself. 
Furthermore, by the definition of~$\tau$ we have~$\tau(E(G))=X$. 
Thus,~$(\mathcal{G},G,\varphi)$ is an~$X$-embedding according to Definition~\ref{def:embedding}. 
Then, by Lemma~\ref{lem:embedding} we have that~$X$ is a yes-instance of PCA if and only if there is a time step~$t$ in which~$G$ exists in~$\mathcal{G}(t)$.
In other words:~$X$ is a yes instance of PCA if and only if~$G$ is isomorphic to~$\mathcal{G}(t)$ with the identity.
It remains to show that~$G$ is isomorphic to~$\mathcal{G}(t)$ if and only if~$\mathcal{G}(t)$ contains~$H$ as a minor.
Suppose that~$G$ is isomorphic to~$\mathcal{G}(t)$.
Hence, each edge in~$E(G)$ exists.
Contracting the edges~$\{v_i,v_{i+1}\}$ for each~$i\in[1,k]$ leads to the graph~$H$.
Conversely, suppose that~$G$ is not isomorphic to~$\mathcal{G}(t)$.
Hence, one of the edges~$\{v_i,v_{i+1}\}$ for some~$i\in[1,k]$ is not present at this time step~$t$.
Thus, the vertices~$V(C)\cup\{v_i\mid i\in[2,k]\}$ do not form a cycle anymore.
Hence, the number of cycles of~$\mathcal{G}(t)$ is lower than the number of cycles of~$H$ and thus~$H$ is not a minor of~$\mathcal{G}(t)$. 
\end{proof}
\end{toappendix}

Since Theorem~\ref{thm-minor-containment-hard-triangle} shows hardness for each fixed minor containing a cycle, it remains to consider fixed minors~$H$ which are forests.
Second, we provide NP-hardness for forests containing a tree with some minimum-degree vertices.

\begin{restatable}{theorem}{restMinorHardStar}
\Minor is \NP-complete and~\W{1}-hard parameterized by $|G|$ for every fixed forest~$H$ with a connected component that contains $a)$ at least~$2$ vertices of degree at least~$3$ or $b)$ one vertex of degree at least~$4$.
	\label{thm-minor-containment-hard-tree-many-high}
\end{restatable}
\begin{toappendix}
\begin{proof}
For both cases~$a)$ and~$b)$ we present a reduction from PCA.
Let~$X$ be an instance of PCA. 
Without loss of generality, we may assume that~$1 \in X$, since otherwise, we may replace~$X$ by the equivalent instance~$\widetilde{X}:=X \cup \{1\}$. 
Thus, let~$X:=\{x_1, \dots, x_k,1\}$.

First, we show~$a)$.
 We define an equivalent instance~$(\mathcal{G},H)$ of \Minor{} as follows: 
Let~$H$ be a fixed forest such that at least one connected component of~$H$ contains at least~$2$ vertices~$u$ and~$w$ of degree at least~$3$. 
Next, we let~$e$ be an arbitrary but fixed edge on the unique path from~$u$ to~$w$ in~$H$.
Now, we construct the underlying graph~$G$ of~$\mathcal{G}$.
The graph~$G$ consists of a copy of~$H$ in which the edge~$e$ is subdivided exactly~$k-1$ times, that is, we replace~$e=\{v_1,v_{k+1}\}$ by the path~$v_1,v_2,\ldots, v_{k},v_{k+1}$ of~$k-1$ new vertices.
In the following, we call a vertex pair~$\{x,y\}$ of a graph~$F$ \emph{important}, if both~$x$ and~$y$ have degree at least~$3$ in~$F$.
Observe that the number of connected important pairs in~$H$ is equal to the number of connected important pairs of~$G$.
It remains to define~$\tau$.
We set~$\tau(\{v_i,v_{i+1}\}):=x_i$ for every~$i\in [1,k]$.
For every remaining edge~$e'$ of~$E(G)$ we set~$\tau(e'):=1$.
This completes our construction.

Now, we show the correctness of our construction.
We use Lemma~\ref{lem:embedding} to show that~$X$ is a yes-instance of PCA if and only if~$(\mathcal{G},H)$ is a yes-instance of \Minor. 
Observe that~$\varphi: V(G) \rightarrow V(G)$ with~$\varphi(v):=v$ is a trivial monomorphism that identifies~$G$ with itself. 
Furthermore, by the definition of~$\tau$ we have~$\tau(E(G))=X$. 
Thus,~$(\mathcal{G},G,\varphi)$ is an~$X$-embedding according to Definition~\ref{def:embedding}. 
Then, by Lemma~\ref{lem:embedding} we have that~$X$ is a yes-instance of PCA if and only if there is a time step~$t$ in which~$G$ exists in~$\mathcal{G}(t)$.
In other words:~$X$ is a yes instance of PCA if and only if~$G$ is isomorphic to~$\mathcal{G}(t)$ with the identity.
It remains to show that~$G$ is isomorphic to~$\mathcal{G}(t)$ if and only if~$\mathcal{G}(t)$ contains~$H$ as a minor.
Suppose that~$G$ is isomorphic to~$\mathcal{G}(t)$.
Hence, each edge in~$E(G)$ exists.
Contracting the edges~$\{v_i,v_{i+1}\}$ for each~$i\in[1,k]$ leads to the graph~$H$.
Conversely, suppose that~$G$ is not isomorphic to~$\mathcal{G}(t)$.
Hence, one of the edges~$\{v_i,v_{i+1}\}$ for some~$i\in[1,k]$ is not present at this time step~$t$.
Since~$G$ does not contain any cycle, we thus conclude that the vertices~$v_1$ and~$v_{k+1}$ are not in the same connected component in~$\mathcal{G}(t)$.
Hence, the total number of connected important vertex pairs in~$\mathcal{G}(t)$ is lower than the total number of connected important vertex pairs in~$H$ and thus~$H$ is not a minor of~$\mathcal{G}(t)$. 

Second, we show~$b)$.
 We define an equivalent instance~$(\mathcal{G},H)$ of \Minor{} as follows: 
Here, we assume that~$H$ is a fixed forest such that no connected component of~$H$ contains at least two vertices of degree at least~$3$ and at least one connected component of~$H$ contains a vertex of degree at least~$4$. 
Let~$v$ be a vertex of \emph{maximum degree} in~$H$ and let~$\ell$ denote the number of vertices of~$H$ of degree exactly~$\deg_H(v)$. 
Next, we define an auxiliary graph~$H'$ which is then used to define the underlying graph~$G$ of~$\mathcal{G}$.
$H'$ is obtained from~$H$ by replacing the vertex~$v$ with two new vertices~$v_1$ and~$v_{k+1}$ which are adjacent.
Let~$Z:= N_H(v)$ and let~$Y$ be an arbitrary subset of size exactly~$2$ of~$Z$.
For each vertex~$y\in Y$ we add the edge~$\{v_1,y\}$ to~$H'$ and for each~$z\in Z\setminus Y$ we add the edge~$\{v_{k+1},z\}$ to~$H'$.
Note that~$\deg_{H'}(v_1)=3$ and~$\deg_{H'}(v_{k+1})\ge 3$ since by assumption~$|Z|\ge 4$.
Now, we construct the underlying graph~$G$ of~$\mathcal{G}$.
The graph~$G$ consists of a copy of~$H'$ in which the edge~$\{v_1,v_{k+1}\}$ is subdivided exactly~$k-1$ times, that is, we replace~$\{v_1,v_{k+1}\}$ by the path~$v_1,v_2,\ldots, v_{k},v_{k+1}$ of~$k-1$ new vertices.
Furthermore, note that the number of vertices of degree exactly~$\deg_H(v)$ in~$G$ is exactly~$\ell-1$.
It remains to define~$\tau$.
We set~$\tau(\{v_i,v_{i+1}\}):=x_i$ for every~$i\in [1,k]$.
For every remaining edge~$e'$ of~$E(G)$ we set~$\tau(e'):=1$.
This completes our construction.

Now, we show the correctness of our construction.
We use Lemma~\ref{lem:embedding} to show that~$X$ is a yes-instance of PCA if and only if~$(\mathcal{G},H)$ is a yes-instance of \Minor. 
Observe that~$\varphi: V(G) \rightarrow V(G)$ with~$\varphi(v):=v$ is a trivial monomorphism that identifies~$G$ with itself. 
Furthermore, by the definition of~$\tau$ we have~$\tau(E(G))=X$. 
Thus,~$(\mathcal{G},G,\varphi)$ is an~$X$-embedding according to Definition~\ref{def:embedding}. 
Then, by Lemma~\ref{lem:embedding} we have that~$X$ is a yes-instance of PCA if and only if there is a time step~$t$ in which~$G$ exists in~$\mathcal{G}(t)$.
In other words:~$X$ is a yes instance of PCA if and only if~$G$ is isomorphic to~$\mathcal{G}(t)$ with the identity.
It remains to show that~$G$ is isomorphic to~$\mathcal{G}(t)$ if and only if~$\mathcal{G}(t)$ contains~$H$ as a minor.
Suppose that~$G$ is isomorphic to~$\mathcal{G}(t)$.
Hence, each edge in~$E(G)$ exists.
Contracting the edges~$\{v_i,v_{i+1}\}$ for each~$i\in[1,k]$ leads to the graph~$H$.
Conversely, suppose that~$G$ is not isomorphic to~$\mathcal{G}(t)$.
Hence, one of the edges~$\{v_i,v_{i+1}\}$ for some~$i\in[1,k]$ is not present at this time step~$t$.

To show that~$H$ is no minor of~$\mathcal{G}(t)$ we rely on the following observation:
Contracting any adjacent vertices~$x$ and~$y$ of degree~$d_x$ and~$d_y$ into a new vertex~$z$ in some graph~$F$ may only leads to a degree~$d_z>\max(d_x,d_y)$ if both~$d_x$ and~$d_y$ are at least~$3$.
Now, observe that by our assumption each connected component of~$H$ does not contains two vertices of degree~$3$ and at least one connected component of~$H$ contains a vertex of degree at least~$4$.
Hence, each connected component of~$H$ contains at most one vertex of degree at least~$3$.
According to our construction of~$G$, the underlying graph of~$\mathcal{G}(t)$, the unique connected component containing at least two vertices of degree~$3$ is the connected component containing vertices~$v_1$ and~$v_{k+1}$.
Now, since~$\{v_i,v_{i+1}\}$ is not present at time step~$t$ and since~$G$ is a forest, we conclude that~$\mathcal{G}(t)$ has no connected component with at least two vertices of degree at least~$3$.
Hence, by the above observation, no series of edge contractions can increase the maximum degree of any connected component of~$\mathcal{G}(t)$.
Now, since the number of vertices of degree~$\deg_H(v)$ in~$\mathcal{G}(t)$ is~$\ell-1$, we conclude that~$H$ is no minor of~$\mathcal{G}(t)$ since the number of vertices of degree~$\deg_H(v)$ in~$H$ is~$\ell$. 
\end{proof}
\end{toappendix}

For all remaining cases, that is, each connected component of~$H$ is  either a path, or a tree with exactly one vertex of degree~$3$ and no vertex of degree at least~$4$, we provide a polynomial time algorithm.
More precisely, we present an XP-algorithm for the parameter~$h$, the number of vertices of~$H$.
The algorithm works completely analogue to the algorithm of Theorem~\ref{thm:SGXPh} for \SG.
This algorithm also works for minors, since for this structure of~$H$ the minor must already be contained as a subgraph.

\begin{corollary}
The \Minor problem can be solved in $\Oh(n^{h} \cdot \max(L_\mathcal{G})^{(h^2)}) \cdot 2^{\Oh(\sqrt{h \log h})}$~time if~$H$ is a forest such that each connected component of~$H$ contains no vertex of degree at least~$4$ and at most one vertex of degree~$3$.
\label{cor-minor-containment-p}
\end{corollary}

Finally, we take a closer look on the parameterized complexity of the~$4$ problems concerning minors and subgraphs.


\begin{restatable}{corollary}{restCollectionHardGandNumber}
	The problems \SG, \SGfree, \Minor, and \Minorfree are
	\begin{itemize}
		\item \NP-hard even if~$G$ is a disjoint union of paths and~$\#1_\text{max}\in\Oh(1)$ and
		\item \NP-hard even if~$G$ is a disjoint union of paths and~$\#0_\text{max}\in\Oh(1)$.
	\end{itemize}
\end{restatable}
\begin{toappendix}
\begin{proof}
The proofs of~\Cref{lem-subgrpah-free-np-h-2} and~\Cref{thm-minor-free-np-h} show reductions from~\PCA to~\SGfree and~\Minorfree respectively, for each fixed~$H$ containing at least one edge.
Moreover, for an instance~$X$ of~\PCA, the underlying graph~$G$ of the constructed EPG is a disjoint union of~$|X|$ copies of~$H$ and the set of labels of the edges of the constructed EPG is the set~$\overline{X} := \{\overline{x} \mid x\in X\}$, where~$\overline{x}$ is obtained from~$x$ by replacing each 0 by a 1 and vice versa.
By~\Cref{PCA const number zeroes} and~\Cref{PCA const number ones},~\PCA is~\NP-hard if~$\#1_\text{max} \in \Oh(1)$ and~\NP-hard if~$\#0_\text{max} \in \Oh(1)$.
Thus, for~$H$ being the graph consisting of a single edge,~\Cref{lem-subgrpah-free-np-h-2} and~\Cref{thm-minor-free-np-h} imply~\NP-hardness for~\SGfree and~\Minorfree even if~$G$ is a matching and~$\#1_\text{max} \in \Oh(1)$ or~$\#0_\text{max} \in \Oh(1)$.

The proof of \Cref{thm:SGNP-h} shows~\NP-hardness for~\SG even if~$H=G$ is a path and where the set of edge labels is exactly the set of strings~$X$ from the~\PCA-instance.
By the above, this implies~\NP-hardness for~\SG even if~$G$ is a path and~$\#1_\text{max} \in \Oh(1)$ or~$\#0_\text{max} \in \Oh(1)$.

Note that this also holds for~\Minor, since~$H=G$ and thus,~$\mathcal{G}(t)$ contains~$H$ as a minor if and only if~$\mathcal{G}(t)$ contains~$H$ as an induced subgraph.
\end{proof}
\end{toappendix}

\begin{restatable}{theorem}{restCollectionHardVertexCoverNumber}
	The problems \SG, \SGfree, \Minor, and \Minorfree are \W1-hard when parameterized by the vertex cover number of the underlying graph even if~$\#1_\text{max} = 1$.
\end{restatable}
\begin{toappendix}
\begin{proof}
We reduce from~\PCA which is~\W1-hard when parameterized by~$|X|$.
First, we describe a general reduction for all four problems and afterwards, we show the correctness for them individually.

Let~$X = \{x_1, \dots, x_{|X|}\}$ be an instance of~\PCA.
For each~$x_i\in X$, let~$J_i := \{j\in[0, |x_i|-1]\mid x_i[j] = 1\}$ denote the positions of 1's of~$x_i$.
For each~$i\in[1, |X|]$, we define~$X_i := \{x_i^j\mid j\in J_i\}$, as the~\emph{split} of~$x_i$, where~$x_i^j$ has the same length as~$x_i$ and only one 1 at position~$j$.
Note that for~$t\geq 0$,~$\modid{x_i}{t} = 1$ if and only if there is some~$j\in J_i$ with~$\modid{x_i^j}{t} = 1$.
Moreover, for every~$t\geq 0$, there is at most one~$j\in J_i$ with~$\modid{x_i^j}{t} = 1$.

We define an EPG~$\mathcal{G} = (V,E,\tau)$ as follows:
We start with an independent set~$\{v_0, \dots, v_{|X|}\}$ and add for each~$x_i\in X$ and each~$j\in J_i$ a vertex~$v_i^j$ to~$V$ and edges~$\{v_{i-1}, v_i^j\}$ and~$\{v_i^j, v_i\}$ to~$E$, both with label~$x_i^j$.
This completes the construction of the EPG.
Note that~$\{v_0, \dots, v_{|X|}\}$ is a vertex cover of size~$|X| + 1$ of the underlying graph~$G$ and that each edge label contains exactly one 1. 

Next, we show that there is a time step~$t$ in which~$\mathcal{G}(t)$ contains a path of~$2\cdot|X|$ edges as an induced subgraph if and only if~$\modid{x_i}{t} = 1$ for each~$x_i\in X$.

$(\Rightarrow)$
If~$\mathcal{G}(t)$ contains a path of~$2\cdot|X|$ edges as an induced subgraph, then by construction and the fact that for any~$x_i\in X$ and any~$j\in J_i$, the edge labels for all incident edges of~$v_i^j$ are equal, for each~$x_i\in X$, there is some~$j\in J_i$ with~$\modid{\tau(\{v_{i-1}, v_i^j\})}{t} = \modid{x_i^j}{t} = 1$.
Hence,~$\modid{x_i}{t} = 1$ for each~$x_i\in X$ and thus~$X$ is a yes-instance of~\PCA.

$(\Leftarrow)$
Suppose that~$\modid{x_i}{t} = 1$ for each~$x_i\in X$.
Then, by the above, for each~$x_i\in X$, there is some~$j_i\in J_i$ with~$\modid{\tau(\{v_{i-1}, v_i^{j_i}\})}{t} = \modid{x_i^{j_i}}{t} = 1$ and $\modid{\tau(\{v_{i-1}, v_i^k\})}{t} = \modid{x_i^k}{t} = 0$ for each~$k\in J_i$ distinct from~$j_i$.
Hence, $(v_0, v_1^{j_1}, v_1, \dots, v_{|X|}^{j_{|X|}})$ is an induced path of~$2\cdot|X|$ edges in~$\mathcal{G}(t)$.

This shows the stated hardness for~\SG.
Next, we argue that this reduction also works for~\Minor.
Recall that for each~$x_i\in X$ and each~$t$, there is at most one~$j\in J_i$ such that~$\modid{x_i^j}{t} = 1$.
Hence, if for some time step~$t$, the graph~$\mathcal{G}(t)$ contains no induced path of~$2\cdot|X|$ edges, then~$\mathcal{G}(t)$ contains at most~$2\cdot|X| - 2$ edges in total, which implies that~$\mathcal{G}(t)$ contains no path of~$2\cdot|X|$ as a minor.
Thus, the stated hardness for~\Minor follows.

It remains to show that the hardness also holds for~\SGfree and~\Minorfree.
To this end, the only thing we have to change is that we initially do not use the split of the strings of~$X$ as the edge labels, but the splits of the complement strings of~$X$.
That is, let~$X'$ be an instance of~\PCA, we define~$X := \{x_i := \overline{x_i'}\mid x_i'\in X'\}$.
Moreover, let again~$J_i$ denote the set of positions where~$x_i$ has a 1 and let for~$j\in J_i$,~$x_i^j$ denote the string having length equal to the length of~$x_i$ and a single 1 at position~$j$.
Thus, asking if there is some~$t$ such that for each~$x_i'\in X'$,~$\modid{x_i'}{t} = 1$ is equivalent to asking if~$\modid{x_i}{t} = 0$ which is further equivalent to asking if~$\modid{x_i^j}{t} = 0$ for each~$j\in J_i$.
Hence,~$X$ is a yes-instance of~\PCA if and only if there is a time step~$t$ where~$\mathcal{G}(t)$ is edgeless.
As a consequence, the stated hardness also follows for~\SGfree and~\Minorfree.
\end{proof}
\end{toappendix}

\begin{theorem}
	\label{thm:collectionFPT}
	The problems  \SG, \SGfree, \Minor, and \Minorfree are \FPT with respect to the combined parameter~$\min(\#1_\text{max},\#0_\text{max})$ plus the number of vertices $|V|$ of $G$.
\end{theorem}

\begin{proof}
We prove the theorem by providing a class of FPT-algorithms solving the four considered problems. Intuitively, our algorithms iterate over all possible graphs that can be present in some time step and check, whether these graphs (not) contain~$H$ as an induced subgraph or as a minor, respectively. To this end, we introduce an auxiliary problem that asks whether there exists a time step where~$\mathcal{G}(t)$ consists of a specific edge set.

\prob{\textsc{EPG Present Edges}}{An~EPG~$\mathcal{G}= (V,E,\tau)$ and an edge set~$E' \subseteq E$}{Is there a time step~$t$ such that~$\mathcal{G}(t)=(V,E')$?}

\begin{restatable}{myclaim}{refClaimPresentEdges} 
\textsc{EPG Present Edges} is FPT for~$|V|+\min(\#1_\text{max},\#0_\text{max})$.
\end{restatable}
We prove the claim by providing a parameterized reduction to PCA parameterized by the total number of runs of 1's, that is, the number of groups of consecutive 1's, in all strings, which is known to be FPT~\cite{DBLP:conf/mfcs/MorawietzRW20}.

\begin{toappendix}
\begin{claimproof} We prove the claim by providing a parameterized reduction to PCA parameterized by the total number of runs of 1's, i.e., the number of groups of consecutive 1's, in all strings, which is known to be FPT~\cite{DBLP:conf/mfcs/MorawietzRW20}.

Let~$(\mathcal{G}=(V,E,\tau), E')$ be an instance of \textsc{EPG Present Edges}. We define~$X:= \{x_e \mid e \in E\}$ by setting~$x_e:= \tau(e)$ for all~$e \in E'$ and~$x_e:=\overline{\tau(e)}$ for all~$e \in E \setminus E'$. Recall that given a string~$s$, the string~$\overline{s}$ results from~$s$ by converting all~$1$'s into~$0$'s and vice versa. 

We first show that the total number of runs of 1's in all strings in~$X$ is bounded by some function in our parameter~$|V|+\min(\#1_\text{max},\#0_\text{max})$. To this end, we show that the total number of runs of 1's is~$a)$ bounded by a function in~$|V|+\#1_\text{max}$ and~$b)$ bounded by a function in~$|V|+\#0_\text{max}$. We first show~$a)$: Let~$s$ be a string with~$\ell$ runs of 1's. Then, $\overline{s}$ has at most~$\ell+1$ runs of~1's. Thus, the total number of runs of 1's in~$X$ is upper-bounded by~$|E| \cdot \#1_\text{max}$ plus one run for each edge~$e\in E \setminus E'$, which is at most~$|E| \cdot (\#1_\text{max}+1)$. We next show~$b)$: Let~$s$ be a string with~$\ell$ runs of 1's. Then,~$\overline{s}$ has~$\ell$ runs of 0's and~$s$ has at most~$\ell+1$ runs of~$0$'s. Thus, the total number of runs of 1's in~$X$ is bounded by~$|E| \cdot \#0_\text{max}$ plus one run for each edge~$e\in E'$, which is at most~$|E| \cdot (\#0_\text{max}+1)$.

We complete the proof of the claim by showing that the reduction is correct. Clearly,~$(\mathcal{G},E')$ is a yes-instance of \textsc{EPG Present Edges} if and only if there is a time step~$t$ such that~$\tau(e)[t]^\circ=1$ for all~$e \in E'$ and~$\tau(e)[t]^\circ=0$ for all~$e \in E \setminus E'$. By the construction of~$X$, this is equivalent to the fact that~$x_e[t]^\circ=1$ for all~$x_e \in X$. Therefore,~$(\mathcal{G},E')$ is a yes-instance of \textsc{EPG Present Edges} if and only if~$X$ is a yes-instance of PCA.
\end{claimproof}
\end{toappendix}

We next describe the FPT-algorithms for \SG, \SGfree, \Minor, and \Minorfree. Let~$(\mathcal{G}=(V,E,\tau),H)$ be an instance of one of these problems. We iterate over every possible~$E' \subseteq E$ and check if~$(V,E')$ (not) contains an induced~$H$ or (not) contains~$H$ as a minor, respectively. If this is the case, we check whether~$(\mathcal{G},E')$ is a yes-instance of \textsc{EPG Present Edges} and return \emph{yes} or \emph{no} accordingly.

The correctness follows by the fact that we consider every possible graph $(V,E')$ that might be present in some time step. It remains to consider the running time. Due to the previous claim, checking whether~$(\mathcal{G},E')$ is a yes-instance of \textsc{EPG Present Edges} can be performed in FPT time parameterized by~$|V|+\min(\#1_\text{max},\#0_\text{max})$. Checking whether~$H$ is a minor of~$(V,E')$ or~$H$ is an induced subgraph of~$(V,E')$ can clearly be done in a running time only depending on the graph size~$|V|$. Consequently, the four considered problems are~FPT when parameterized by~$|V|+\min(\#1_\text{max},\#0_\text{max})$. 
\end{proof}

\nocite{rosser1941explicit,schaefer2002completeness,DBLP:journals/tcs/Stockmeyer76,DBLP:journals/tcs/Wrathall76,CliqueReduction}

\bibliographystyle{plain} 
\bibliography{bib-slim}

\begin{thebibliography}{10}

\bibitem{DBLP:journals/jcss/AkridaMNRSZ20}
Eleni~C. Akrida, George~B. Mertzios, Sotiris~E. Nikoletseas, Christoforos~L.
  Raptopoulos, Paul~G. Spirakis, and Viktor Zamaraev.
\newblock How fast can we reach a target vertex in stochastic temporal graphs?
\newblock {\em Journal of Computer and System Sciences}, 114:65--83, 2020.

\bibitem{BL83}
L{\'{a}}szl{\'{o}} Babai and Eugene~M. Luks.
\newblock Canonical labeling of graphs.
\newblock In David~S. Johnson, Ronald Fagin, Michael~L. Fredman, David Harel,
  Richard~M. Karp, Nancy~A. Lynch, Christos~H. Papadimitriou, Ronald~L. Rivest,
  Walter~L. Ruzzo, and Joel~I. Seiferas, editors, {\em Proceedings of the 15th
  Annual {ACM} Symposium on Theory of Computing, 25-27 April, 1983, Boston,
  Massachusetts, {USA}}, pages 171--183. {ACM}, 1983.

\bibitem{DBLP:journals/networks/Berman96}
Kenneth~A. Berman.
\newblock Vulnerability of scheduled networks and a generalization of
  {M}enger's theorem.
\newblock {\em Networks}, 28(3):125--134, 1996.

\bibitem{bhadra2003complexity}
Sandeep Bhadra and Afonso Ferreira.
\newblock Complexity of connected components in evolving graphs and the
  computation of multicast trees in dynamic networks.
\newblock In Samuel Pierre, Michel Barbeau, and Evangelos Kranakis, editors,
  {\em Ad-Hoc, Mobile, and Wireless Networks, Second International Conference,
  {ADHOC-NOW} 2003 Montreal, Canada, October 8-10, 2003, Proceedings}, volume
  2865 of {\em Lecture Notes in Computer Science}, pages 259--270. Springer,
  2003.

\bibitem{casteigts:hal-00865762}
Arnaud Casteigts and Paola Flocchini.
\newblock Deterministic algorithms in dynamic networks: {F}ormal models and
  metrics.
\newblock Technical report, 2013.

\bibitem{casteigts:hal-00865764}
Arnaud Casteigts and Paola Flocchini.
\newblock Deterministic algorithms in dynamic networks: {P}roblems, analysis,
  and algorithmic tools.
\newblock Technical report, 2013.

\bibitem{DBLP:journals/paapp/CasteigtsFQS12}
Arnaud Casteigts, Paola Flocchini, Walter Quattrociocchi, and Nicola Santoro.
\newblock Time-varying graphs and dynamic networks.
\newblock {\em International Journal of Parallel, Emergent and Distributed
  Systems}, 27(5):387--408, 2012.

\bibitem{DBLP:books/sp/CyganFKLMPPS15}
Marek Cygan, Fedor~V. Fomin, Lukasz Kowalik, Daniel Lokshtanov, D{\'{a}}niel
  Marx, Marcin Pilipczuk, Michal Pilipczuk, and Saket Saurabh.
\newblock {\em {Parameterized Algorithms}}.
\newblock Springer, 2015.

\bibitem{dRLP98}
Willem~P. de~Roever, Hans Langmaack, and Amir Pnueli, editors.
\newblock {\em Compositionality: {T}he Significant Difference, International
  Symposium, {COMPOS}'97. Revised Lectures}, volume 1536 of {\em Lecture Notes
  in Computer Science}. Springer, 1998.

\bibitem{DBLP:conf/edbt/DingYQ08}
Bolin Ding, Jeffrey~Xu Yu, and Lu~Qin.
\newblock Finding time-dependent shortest paths over large graphs.
\newblock In Alfons Kemper, Patrick Valduriez, Noureddine Mouaddib, Jens
  Teubner, Mokrane Bouzeghoub, Volker Markl, Laurent Amsaleg, and Ioana
  Manolescu, editors, {\em Proceedings of the 11th International Conference on
  Extending Database Technology}, volume 261 of {\em {ACM} International
  Conference Proceeding Series}, pages 205--216. {ACM}, 2008.

\bibitem{DBLP:series/mcs/DowneyF99}
Rodney~G. Downey and Michael~R. Fellows.
\newblock {\em {Parameterized Complexity}}.
\newblock Monographs in Computer Science. Springer, 1999.

\bibitem{erlebach2020game}
Thomas Erlebach and Jakob~T. Spooner.
\newblock A game of cops and robbers on graphs with periodic edge-connectivity.
\newblock In Alexander Chatzigeorgiou, Riccardo Dondi, Herodotos Herodotou,
  Christos~A. Kapoutsis, Yannis Manolopoulos, George~A. Papadopoulos, and
  Florian Sikora, editors, {\em {SOFSEM} 2020: Theory and Practice of Computer
  Science - 46th International Conference on Current Trends in Theory and
  Practice of Informatics, {SOFSEM} 2020, Limassol, Cyprus, January 20-24,
  2020, Proceedings}, volume 12011 of {\em Lecture Notes in Computer Science},
  pages 64--75. Springer, 2020.

\bibitem{ganguly2009dynamics}
Niloy Ganguly, Andreas Deutsch, and Animesh Mukherjee.
\newblock Dynamics on and of complex networks.
\newblock {\em Applications to Biology, Computer Science, and the Social
  Sciences}, 2009.

\bibitem{holme2015modern}
Petter Holme.
\newblock Modern temporal network theory: {A} colloquium.
\newblock {\em The European Physical Journal B}, 88(9):1--30, 2015.

\bibitem{holme2012temporal}
Petter Holme and Jari Saram{\"a}ki.
\newblock Temporal networks.
\newblock {\em Physics reports}, 519(3):97--125, 2012.

\bibitem{DBLP:conf/concur/JeckerM021}
Isma{\"{e}}l Jecker, Nicolas Mazzocchi, and Petra Wolf.
\newblock Decomposing permutation automata.
\newblock In Serge Haddad and Daniele Varacca, editors, {\em 32nd International
  Conference on Concurrency Theory, {CONCUR} 2021, August 24-27, 2021, Virtual
  Conference}, volume 203 of {\em LIPIcs}, pages 18:1--18:19. Schloss Dagstuhl
  - Leibniz-Zentrum f{\"{u}}r Informatik, 2021.

\bibitem{DBLP:journals/jct/KawarabayashiKR12}
Ken{-}ichi Kawarabayashi, Yusuke Kobayashi, and Bruce~A. Reed.
\newblock The disjoint paths problem in quadratic time.
\newblock {\em Journal of Combinatorial Theory, Series {B}}, 102(2):424--435,
  2012.

\bibitem{DBLP:journals/jcss/KempeKK02}
David Kempe, Jon~M. Kleinberg, and Amit Kumar.
\newblock Connectivity and inference problems for temporal networks.
\newblock {\em Journal of Computer and System Sciences}, 64(4):820--842, 2002.

\bibitem{leskovec2014snap}
Jure Leskovec and Andrej Krevl.
\newblock {SNAP} datasets: {S}tanford large network dataset collection.
\newblock \url{http://snap.stanford.edu/data/}, 2014.

\bibitem{lovasz2006graph}
L{\'a}szl{\'o} Lov{\'a}sz.
\newblock Graph minor theory.
\newblock {\em Bulletin of the American Mathematical Society}, 43(1):75--86,
  2006.

\bibitem{DBLP:journals/algorithmica/MertziosMS19}
George~B. Mertzios, Othon Michail, and Paul~G. Spirakis.
\newblock Temporal network optimization subject to connectivity constraints.
\newblock {\em Algorithmica}, 81(4):1416--1449, 2019.

\bibitem{DBLP:journals/tcs/MichailS16}
Othon Michail and Paul~G. Spirakis.
\newblock Traveling salesman problems in temporal graphs.
\newblock {\em Theoretical Computer Science}, 634:1--23, 2016.

\bibitem{DBLP:journals/cacm/MichailS18}
Othon Michail and Paul~G. Spirakis.
\newblock Elements of the theory of dynamic networks.
\newblock {\em Communications of the {ACM}}, 61(2):72, 2018.

\bibitem{DBLP:conf/mfcs/MorawietzRW20}
Nils Morawietz, Carolin Rehs, and Mathias Weller.
\newblock A timecop's work is harder than you think.
\newblock In Javier Esparza and Daniel Kr{\'{a}}l', editors, {\em 45th
  International Symposium on Mathematical Foundations of Computer Science,
  {MFCS} 2020, August 24-28, 2020, Prague, Czech Republic}, volume 170 of {\em
  LIPIcs}, pages 71:1--71:14. Schloss Dagstuhl - Leibniz-Zentrum f{\"{u}}r
  Informatik, 2020.

\bibitem{DBLP:conf/mfcs/Morawietz021}
Nils Morawietz and Petra Wolf.
\newblock A timecop's chase around the table.
\newblock In Filippo Bonchi and Simon~J. Puglisi, editors, {\em 46th
  International Symposium on Mathematical Foundations of Computer Science,
  {MFCS} 2021, August 23-27, 2021, Tallinn, Estonia}, volume 202 of {\em
  LIPIcs}, pages 77:1--77:18. Schloss Dagstuhl - Leibniz-Zentrum f{\"{u}}r
  Informatik, 2021.

\bibitem{DBLP:journals/rsl/Pagin21}
Peter Pagin.
\newblock Compositionality, computability, and complexity.
\newblock {\em The Review of Symbolic Logic}, 14(3):551--591, 2021.

\bibitem{CliqueReduction}
OpenDSA Project.
\newblock Reduction of 3-sat to clique.
\newblock
  \url{https://opendsa-server.cs.vt.edu/ODSA/Books/Everything/html/threeSAT_to_clique.html},
  2021.

\bibitem{RobertsonS95b}
Neil Robertson and Paul~D. Seymour.
\newblock Graph {M}inors .{XIII}. {T}he {D}isjoint {P}aths {P}roblem.
\newblock {\em Journal of Combinatorial Theory, Series B}, 63(1):65--110, 1995.

\bibitem{rosser1941explicit}
Barkley Rosser.
\newblock Explicit bounds for some functions of prime numbers.
\newblock {\em American Journal of Mathematics}, 63(1):211--232, 1941.

\bibitem{10.1371/journal.pone.0130824}
Piotr Sapiezynski, Arkadiusz Stopczynski, Radu Gatej, and Sune Lehmann.
\newblock Tracking human mobility using {WiFi} signals.
\newblock {\em PLOS ONE}, 10(7):1--11, 07 2015.

\bibitem{schaefer2002completeness}
Marcus Schaefer and Christopher Umans.
\newblock Completeness in the polynomial-time hierarchy: {A} compendium.
\newblock {\em SIGACT news}, 33(3):32--49, 2002.

\bibitem{DBLP:journals/tcs/Stockmeyer76}
Larry~J. Stockmeyer.
\newblock The polynomial-time hierarchy.
\newblock {\em Theoretical Computer Science}, 3(1):1--22, 1976.

\bibitem{wehmuth2015unifying}
Klaus Wehmuth, Artur Ziviani, and Eric Fleury.
\newblock A unifying model for representing time-varying graphs.
\newblock In {\em 2015 {IEEE} International Conference on Data Science and
  Advanced Analytics, {DSAA} 2015, Campus des Cordeliers, Paris, France,
  October 19-21, 2015}, pages 1--10. {IEEE}, 2015.

\bibitem{DBLP:journals/tcs/Wrathall76}
Celia Wrathall.
\newblock Complete sets and the polynomial-time hierarchy.
\newblock {\em Theoretical Computer Science}, 3(1):23--33, 1976.

\bibitem{DBLP:journals/pvldb/WuCHKLX14}
Huanhuan Wu, James Cheng, Silu Huang, Yiping Ke, Yi~Lu, and Yanyan Xu.
\newblock Path problems in temporal graphs.
\newblock {\em Proceedings of the {VLDB} Endowment}, 7(9):721--732, 2014.

\bibitem{DBLP:journals/comsur/Zhang06}
Zhensheng Zhang.
\newblock Routing in intermittently connected mobile ad hoc networks and delay
  tolerant networks: {O}verview and challenges.
\newblock {\em {IEEE} Communications Surveys and Tutorials}, 8(1-4):24--37,
  2006.

\end{thebibliography}

\end{document}